\documentclass{amsart}
\usepackage{amsmath,amssymb,enumerate,amscd}
\usepackage[dvips]{graphics}
\usepackage[dvips]{graphicx}
\usepackage{epsfig,enumerate,mathrsfs}

\usepackage{color}

\numberwithin{equation}{section}
\newtheorem{theorem}{Theorem}[section]
\newtheorem{proposition}[theorem]{Proposition}
\newtheorem{lemma}[theorem]{Lemma}
\newtheorem{corollary}[theorem]{Corollary}

\def\Sph{\mathbb{S}}
\DeclareMathOperator{\Tr}{Tr}

\newcommand{\R}{\mathbb R}
\newcommand{\Z}{\mathbb Z}
\def\N{\mathbb{N}}
\def\C{\mathbb{C}}

\def\U{\mathcal{U}}

\def\Id{\mathrm{Id}}

\def\D{\partial}
\newcommand{\bzd}{{\det}_{b}}
\newcommand{\bz}{{}^b\! \zeta}

\def\dbar{d\hspace*{-0.08em}\bar{}\hspace*{0.1em}}

\begin{document}

\title[Analytic surgery]
{Analytic surgery of the zeta function}

\author{Klaus Kirsten}
\address{Department of Mathematics\\ Baylor University\\
         Waco\\ TX 76798\\ U.S.A. }
\email{Klaus$\_$Kirsten@baylor.edu}

\author{Paul Loya}
\address{Department of Mathematics \\Binghamton University\\
Vestal Parkway East\\Binghamton\\NY 13902\\ U.S.A. }
\email{paul@math.binghamton.edu}


\date{{\today. file name}: KL-CMP.tex}

\begin{abstract}
In this paper we study the asymptotic behavior (in the sense of meromorphic functions) of the zeta function of a Laplace-type operator on a closed manifold when the underlying manifold is stretched in the direction normal to a dividing hypersurface, separating the manifold into two manifolds with infinite cylindrical ends. We also study the related problem on a manifold with boundary as the manifold is stretched in the direction normal to its boundary, forming a manifold with an infinite cylindrical end. Such singular deformations fall under the category of ``analytic surgery'', developed originally by Hassell, Mazzeo and Melrose \cite{mazz95-5-14,hass95-3-115,hass98-6-255} in the context of eta invariants and determinants.
\end{abstract}

\maketitle


\section{Introduction}

The behavior of global spectral invariants of Laplace-type operators under various deformations plays an important role in different areas of mathematics and physics. For example, the behavior of effective actions under conformal transformations has been intensively studied \cite{blau89-4-1467,dett92-377-252,dowk86-33-3150,dowk94-162-633,dowk78-11-895,dowk90-31-808}. The main reason for these studies is that exact results for a given operator may sometimes be obtained by transforming to a simpler operator where the answer is known \cite{bran94-344-479}. This has applications in quantum field theories in curved space times \cite{blau88-209-209,blau89-4-1467,buch85-162-92,buch89-12-1} and in finite temperature theories in static spacetimes \cite{dowk88-38-3327,dowk89-327-267,kirs91-8-2239}. Also the change of the effective action when deforming the boundary of a region can be studied this way \cite{dett92-377-252,dowk94-162-633}.

Mathematically the analysis of effective actions amounts to the evaluation of functional determinants as they have been introduced by Ray and Singer \cite{ray71-7-145} to give a definition of the Reidemeister-Franz torsion \cite{fran35-173-245}. The above mentioned conformal transformation properties have been crucial in the proof of extremal properties of determinants \cite{bran92-149-241,osgo88-80-148}. But also completely different transformation properties have been analyzed. In particular, the behavior of determinants of Laplace-type operators with respect to certain singular deformations has been analyzed in great detail. One type of deformation is a literal \emph{cut and paste}
decomposition formula for the determinant when the underlying manifold is cut along a dividing hypersurface into two manifolds with boundary.  This was initiated by Burghelea, Friedlander and Kappeler \cite{BuD-FrL-KaT92} with further developments in, e.g., \cite{CarG02,hass98-6-255,hass99-18-971,lee03-355-4093,loya04-48-1279,MulJ-MulW06,vish95-167-1}.

The main focus of this paper is on a different type of deformation, called \emph{analytic surgery} (although we do have something to say about ``cutting and pasting'' --- see Section \ref{sec-glueform}). This method was introduced to study the behavior of the eta and functional determinant invariants of Dirac- and Laplace-type operators when a collar neighborhood of a dividing hypersurface of a closed manifold is stretched to a cylinder of infinite length, or when a collar neighborhood of the boundary of a manifold with boundary is stretched to an infinite cylinder. The limit manifolds under analytic surgery are noncompact complete manifolds and there are additional complications due to the presence of a continuous spectrum, which can be addressed using techniques such as Melrose's $b$-calculus \cite{BMeR93}. To our knowledge, analytic surgery was first discussed geometrically, in the context of the eta invariant, by Singer \cite{ISiI88}, and the first papers to systematize the analysis of such degenerations for the eta invariant and functional determinant were provided by Douglas and Wojciechowski \cite{doug91-142-139,WoK94,WoK95} (who named the process \emph{taking the adiabatic limit}) and by Hassell, Mazzeo and Melrose \cite{hass98-6-255,hass95-3-115,mazz95-5-14}, from which we get the terminology \emph{analytic surgery}. Later related developments are given by various authors in \cite{IDaiX06,lee03-355-4093,ILeeY05,loya05-15-285,MulJ-MulW06,PaJ-WoK00,PaJ-WoK02,PaJ-WoK02II,PaJ-WoK06}.
The methods of \cite{doug91-142-139} and \cite{hass98-6-255,hass95-3-115,mazz95-5-14} are quite different. The former is based on heat kernel estimates with a systematic use of the Duhamel principle for the heat kernel. The latter is based on encoding the degeneracy of the Schwartz kernel of the resolvent uniformly as the cylinder is stretched in an appropriate blown-up manifold \cite{BMeR93}. It is this latter method on which the present paper is based, with the exception that a gluing-type formula is used to bypass the surgery calculus of Hassell-Mazzeo-Melrose to directly analyze resolvents using the $b$-calculus \cite{BMeR93}.

A seemingly different area where the change of spectral properties is relevant is the Casimir effect; see, e.g., \cite{bord01-353-1,eliz94b,emig07-99-170403,hert05-95-250402,milt01b,milt04-37-209}. Calculations of Casimir energies are often plagued by divergencies and suitable subtractions need to be made. This need is based on the fact that only energy differences between two states have a physical meaning. By comparing suitable configurations infinities cancel and finite answers are obtained. In this context it would be most desirable to know how the Casimir energy changes when one of these configurations is deformed into the other, where typical deformations would be a change
in the geometry of an object or a boundary condition. The change between configurations in suitable classes would be finite by construction and no ambiguities would arise \cite{scha06-73-042102}.

Several approaches to analyze the Casimir energy are available. Technically closest related to the topic of functional determinants is the zeta function method. Assuming a discrete spectrum $\lambda_1\leq \lambda_2 \leq ...\to\infty$ of the Laplace-type operator $\Delta$, the zeta function is defined by
\[
\zeta (\Delta, s) = \Tr(\Delta^{-s}) = \sum_i \lambda_i^{-s},
\]
where $\Re s $ of the complex parameter $s$ needs to be sufficiently large such as to make this sum convergent \cite{weyl12-71-441}. The determinant is then defined by $(d/ds)|_{s=0} \zeta (s)$, whereas the Casimir energy is related to (the finite part of) $\zeta(s=-1/2)$. If in addition to the above mentioned transformation properties for determinants analogous properties were to exist for the Casimir energies, it would seem natural to assume that suitable relations should not just hold at $s=0$ and $s=-1/2$, but in fact for all values of $s$. This is exactly what the present article is about. Although our initial goal was to find relations between Casimir energies for different configurations, the just mentioned observation led us to analytic surgery formulas for zeta functions valid for all values of the complex parameter $s$. The results obtained are described and summarized in the following.

\subsection{Stretching manifolds with boundary}

Let $\Delta$ be a Laplace-type operator on $M_0$, a compact Riemannian manifold with boundary, and let $Y = \partial M_0$. Throughout this paper, `Laplace-type' means a symmetric (formally self-adjoint), nonnegative, second order differential operator acting between sections of a Hermitian vector bundle whose principal symbol is the underlying Riemannian metric. For notational simplicity we will always leave out vector bundles from our notations and pretend our operators are acting on functions. We assume that $M_0$ has a collar neighborhood
\[
M_0 \cong [-1,0] \times Y
\]
over which $\Delta = -\D_x^2 + \Delta_Y$ where $\Delta_Y$ is a Laplace-type operator on $Y$. Here, we identify the original boundary of $M_0$ with $\{0\} \times Y$. Let $N_r = [0,r] \times Y$ and let $M_r$ be the manifold obtained from $M_0$ by attaching the cylinder $N_r$ to $\partial M_0$,
\[
M_r = M_0 \cup_Y N_r ;
\]
see Figure \ref{fig-glue1}.

\begin{figure}[h!] \centering
\includegraphics{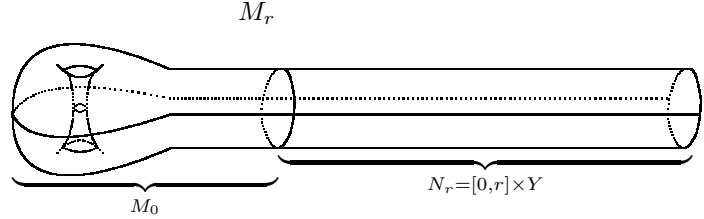} \caption{Sticking the cylinder $N_r = [0,r] \times Y$ onto $M_0$ forms $M_r$.} \label{fig-glue1}
\end{figure}

The Laplace-type operator $\Delta$ has a canonical extension to $M_r$ (as do all the geometric structures on $M_0$) and putting Dirichlet boundary conditions on $\partial M_r = \{r\}\times Y$, we denote the corresponding Dirichlet Laplacian by $\Delta_{M_r}$. We put
\[
M_{\infty} = M_0 \cup_Y \big([0,\infty) \times Y\big),
\]
which is a manifold with cylindrical end. We let $\Delta_{\infty}$ be the canonical extension of $\Delta$ to $M_{\infty}$. Let $\zeta(\Delta_{M_r},s)$ denote the zeta function of the operator $\Delta_{M_r}$ and let $\bz(\Delta_{M_{\infty}}, s)$ denote the $b$-zeta function of $\Delta_{M_{\infty}}$, which was introduced by Piazza \cite{PiP93} and is a natural generalization of the zeta function of compact manifolds to manifolds with cylindrical ends \cite{BMeR93}. An equally natural generalization is the relative zeta function studied in \cite{hass99-18-971,mull98-192-309}. As in \cite{MulJ-MulW06} we put
\[
\xi_Y(s) := \frac{\Gamma\big(s - 1/2\big)}{\sqrt{\pi}\, \Gamma(s)} \, \zeta\big(\Delta_Y,s-1/2\big),
\]
where $\zeta(\Delta_Y,s)$ is the zeta function of $\Delta_Y$. The following is our first result.

\begin{theorem} \label{thm-main0} Assume that $\ker \Delta_Y = \{0\}$ and $\ker \Delta_\infty = \{0\}$. Then for $r \geq r_0$ for some $r_0 > 0$, as meromorphic functions of $s \in \C$ we have
\[
\zeta(\Delta_{M_r}, s) - \frac{r}{2} \xi_Y(s) \equiv \bz(\Delta_{M_{\infty}},s) - \frac14 \zeta(\Delta_Y,s)
\]
modulo an entire function of $s$ that vanishes exponentially fast as $r \to \infty$ uniformly on compact subsets of $\C$. In particular, as meromorphic functions of $s \in \C$,
\[
\lim_{r \to \infty} \left[\zeta(\Delta_{M_r}, s) - \frac{r}{2} \xi_Y(s)\right] = \bz(\Delta_{M_{\infty}},s) - \frac14 \zeta(\Delta_Y,s) .
\]
\end{theorem}

By ``modulo an entire function of $s$ that vanishes exponentially fast as $r \to \infty$ $\ldots$'' we mean that
\[
\zeta(\Delta_{M_r}, s) - \frac{r}{2} \xi_Y(s) = \bz(\Delta_{M_{\infty}},s) - \frac14 \zeta(\Delta_Y,s) + f(r,s),
\]
where $f(r,s) \in C^\infty((r_0 , \infty) \times \C)$ is an entire function of $s \in \C$ such that given any compact subset $K \subseteq \C$ there are constants $c,C > 0$ such that for all $r > r_0$ and $s \in K$,
\[
|f(r,s)| \leq C e^{-cr}.
\]
We remark that Theorem \ref{thm-main0} also holds if $M_0$ has boundary components other than $Y$, but only $Y$ is stretched leaving the other ones fixed; at the other boundary components we put local boundary conditions such as Dirichlet boundary conditions.

Taking the derivative of both sides of the equality in Theorem \ref{thm-main0} and setting $s = 0$, we recover Lee \cite{ILeeY05} and M\"uller and M\"uller's \cite{MulJ-MulW06} analytic surgery formulas for $\zeta$-regularized determinants. (The formulas in \cite{ILeeY05,MulJ-MulW06} were not in terms of $b$-zeta functions but our formula is equivalent to theirs.)

\begin{corollary} With the same assumptions as in Theorem \ref{thm-main0}, we have the following analytic surgery formula for zeta-regularized determinants:
\[
\lim_{r \to \infty} e^{\frac{r}{2} \xi_Y'(0)} \det(\Delta_{M_r}) = (\det(\Delta_{Y} ))^{-1/4}  \, \bzd (\Delta_{M_{\infty}}).
\]
\end{corollary}

\subsection{Stretching closed manifolds along an interior cylinder}

Now let $\Delta$ be a Laplace-type operator on $M$, a closed ($=$ compact without boundary) Riemannian manifold, and let $Y \subseteq M$ be an embedded codimension one hypersurface that divides $M$ into two connected components, the closures of which are smooth manifolds with boundary $M_1$ and $M_2$ with a common boundary $Y := \partial M_1 = \partial M_2$; see Figure \ref{fig-glue}.  We assume that $M$ has a collar neighborhood
\[
M \cong [-1,1] \times Y
\]
over which $\Delta = -\D_x^2 + \Delta_Y$ where $\Delta_Y$ is a Laplace-type operator on $Y$. Here, we identify the original dividing hypersurface with $\{0\} \times Y$.

\begin{figure}[h!] \centering
\includegraphics{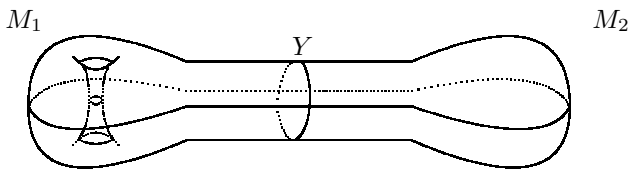} \caption{$M = M_1 \cup_Y M_2$.} \label{fig-glue}
\end{figure}

Now let $N_r = [-r,r] \times Y$ and put
\[
M_r = M_1 \cup_{\{-r\} \times Y} N_r \cup_{\{r\} \times Y} M_2;
\]
in other words, we replace the dividing hypersurface $Y$ in the manifold $M$ by the cylinder $N_r$ and then glue along the ends; see Figure \ref{fig-glue2}.
\begin{figure}[h!] \centering
\includegraphics{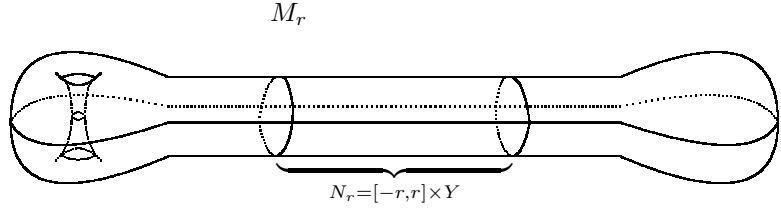} \caption{Replacing the hypersurface $Y$ in $M$ by the collar $N_r = [-r,r] \times Y$ forms the stretched manifold $M_r$.} \label{fig-glue2}
\end{figure}
The Laplace-type operator $\Delta$ has a canonical extension to $M_r$, which we denote by $\Delta_{M_r}$. For $i = 1,2$, we put
\[
M_{i,\infty} = M_i \cup_Y \big([0,\infty) \times Y\big),
\]
which is a manifold with cylindrical end. We let $\Delta_{i,\infty}$ be the canonical extension of $\Delta_{M}|_{M_i}$ to $M_{i,\infty}$. Let $\zeta(\Delta_{M_r},s)$ denote the zeta function of the operator $\Delta_{M_r}$ and $\bz(\Delta_{M_{i,\infty}}, s)$ ($i = 1,2$) denote the $b$-zeta function of $\Delta_{M_{i,\infty}}$. The following is our next result.

\begin{theorem} \label{thm-main} Assume that $\ker \Delta_Y = \{0\}$ and $\ker \Delta_{i,\infty} = \{0\}$, $i = 1,2$. Then for $r \geq r_0$ for some $r_0 > 0$, as meromorphic functions of $s \in \C$ we have
\[
\zeta(\Delta_{M_r}, s) - r \xi_Y(s) \equiv \bz(\Delta_{M_{1,\infty}},s) + \bz(\Delta_{M_{2,\infty}},s)
\]
modulo an entire function of $s$ that vanishes exponentially fast as $r \to \infty$ uniformly on compact subsets of $\C$. In particular, as meromorphic functions of $s \in \C$,
\[
\lim_{r \to \infty} \left[\zeta(\Delta_{M_r}, s) - r \xi_Y(s)\right] = \bz(\Delta_{M_{1,\infty}},s) + \bz(\Delta_{M_{2,\infty}},s).
\]
\end{theorem}

We remark that this theorem holds as stated if $M$ has a boundary as long as $\partial M$ does not intersect $Y$, and at $\partial M$ we put suitable boundary conditions such as Dirichlet boundary conditions.

Taking the derivative of both sides of the equality in Theorem \ref{thm-main} and setting $s = 0$, we recover Lee \cite{ILeeY05} and M\"uller and M\"uller's \cite{MulJ-MulW06} analytic surgery formula for $\zeta$-regularized determinants.

\begin{corollary} With the same assumptions as in Theorem \ref{thm-main}, we have the following analytic surgery formula for zeta-regularized determinants:
\[
\lim_{r \to \infty} e^{r \xi_Y'(0)} \det(\Delta_{M_r}) = \bzd(\Delta_{M_{1,\infty}} ) \cdot \bzd (\Delta_{M_{2,\infty}}).
\]
\end{corollary}

For our last result, let us call $M_{1,r}$ and $M_{2,r}$ the left and right-hand manifolds with boundary obtained by slicing $M_r$ at $\{0\} \times Y$; see Figure \ref{fig-glue0}.
\begin{figure}[h!] \centering
\includegraphics{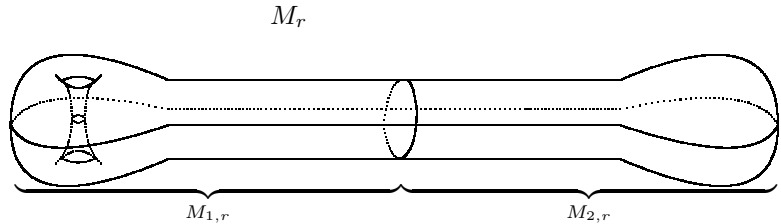} \caption{Separating $M_r$ into $M_{1,r}$ and $M_{2,r}$.} \label{fig-glue0}
\end{figure}
Put Dirichlet boundary conditions at $\{0\} \times Y$ and let $\Delta_{M_{i,r}}$ denote the corresponding Dirichlet Laplacians on $M_{i,r}$ ($i = 1,2$). The following theorem is our final result, which follows trivially from Theorems \ref{thm-main0} and \ref{thm-main}.

\begin{theorem} \label{thm-main1} Assume that $\ker \Delta_Y = \{0\}$ and $\ker \Delta_{i,\infty} = \{0\}$, $i = 1,2$. Then for $r \geq r_0$ for some $r_0 > 0$, as meromorphic functions of $s \in \C$ we have
\[
\zeta(\Delta_{M_r}, s) - \zeta(\Delta_{M_{1,r}}, s) - \zeta(\Delta_{M_{2,r}}, s) \equiv \frac12 \zeta(\Delta_Y,s)
\]
modulo an entire function of $s$ that vanishes exponentially fast as $r \to \infty$ uniformly on compact subsets of $\C$. In particular, as meromorphic functions of $s \in \C$,
\[
\lim_{r \to \infty} \left[\zeta(\Delta_{M_r}, s) - \zeta(\Delta_{M_{1,r}}, s) - \zeta(\Delta_{M_{2,r}}, s)\right] = \frac12 \zeta(\Delta_Y,s).
\]
\end{theorem}

Taking the derivative of both sides of the equality in Theorem \ref{thm-main1} and setting $s = 0$, we recover a particular case of Park and Wojciechowski's adiabatic decomposition formula \cite{PaJ-WoK02II,PaJ-WoK05,PaJ-WoK06}.

\begin{corollary} With the same assumptions as in Theorem \ref{thm-main1}, we have the following analytic surgery formula for zeta-regularized determinants:
\[
\lim_{r \to \infty} \frac{\det(\Delta_{M_r})}{\det(\Delta_{M_{1,r}}) \cdot \det(\Delta_{M_{2,r}})} = \det(\Delta_{Y} )^{1/2}.
\]
\end{corollary}

We now outline this paper. We start in Section \ref{sec-simple} by presenting analytic surgery formulas in the model case of a pure product cylinder, results we will need later.  In the spirit of \cite{BMeR93}, our operators are defined via their Schwartz kernels and for this reason, in Section \ref{sec-trace} we study trace theorems for operators whose Schwartz kernels are continuous (but not necessarily trace-class in the functional analytic sense). Next, in Section \ref{sec-glueform}, following the arguments in \cite{loya04-48-1279} we present a gluing formula for the zeta function when the underlying manifold is cut into two pieces along a dividing hypersurface. We use this gluing formula in Section \ref{sec-anasurg} to prove the analytic surgery theorems in the introduction, modulo some details on Dirichlet-to-Neumann maps which we will present in the Appendix.

\section{Simple examples of analytic surgery} \label{sec-simple}

For pedagogical reasons, before going through the details of our main results, we present these analytic surgery formulas in the simplest possible nontrivial setting, a pure product cylinder,  results we will need later anyhow. This simple situation  illustrates the importance of the invertibility assumptions placed on the cross sectional Laplacians in the main theorems.

\begin{figure}[h!]  \centering
\includegraphics{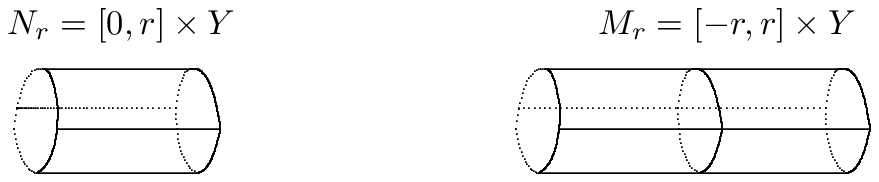} \caption{Product cylinders.} \label{fig-cylinder}
\end{figure}

Let $Y$ be a closed Riemannian manifold and let $\Delta_Y$ be a Laplace-type operator on $Y$, not necessarily invertible; in particular, the spectrum of $\Delta_Y$ consists of nonnegative eigenvalues. For $r > 0$, consider the cylinder $N_r = [0,r]_x \times Y$ (see the left-hand picture in Figure \ref{fig-cylinder}) and the Laplace-type operator
\[
-\partial_x^2 + \Delta_Y
\]
on $N_r$, for which we impose Dirichlet boundary conditions and denote the resulting operator by $\Delta_{N_r}$. In this case we can easily and \emph{explicitly} analyze the $r \to \infty$ behavior of the zeta function on $N_r$. In the following proposition, $\zeta_R(s) := \sum_{n = 1}^\infty n^{-s}$ is the Riemann zeta function.

\begin{proposition} \label{prop-cylinder} As meromorphic functions on $\C$ we have
\begin{equation} \label{zetacyl}
\zeta(\Delta_{N_r},s) = \frac{r}{2} \xi_Y(s) - \frac12 \zeta(\Delta_Y,s)  + \kappa(r,s) + \frac{r^{2s} \dim \ker \Delta_Y}{\pi^{2s}} \,
\zeta_R(2 s) ,
\end{equation}
where $\xi_Y(s) := \frac{\Gamma\big(s - 1/2\big)}{\sqrt{\pi}\, \Gamma(s)} \, \zeta\big(\Delta_Y,s-1/2\big)$ and $\kappa(r,s)$ is an entire function of $s$ that vanishes exponentially fast as $r \to \infty$ uniformly on compact subsets of $\C$; more explicitly, $\kappa(r,s) \in C^\infty((0 , \infty) \times \C)$ and is an entire function of $s \in \C$ such that given any compact subset $K \subseteq \C$ there are constants $c,C > 0$ such that
\[
|\kappa(r,s)| \leq C e^{-cr}\ \ , \ \quad r \geq 1\, , \ \ s \in K.
\]
\end{proposition}
\begin{proof}
In \cite{KiK-LoP-PaJII06} there is a simple proof of the formula for $\zeta(\Delta_{N_r},s)$ using the contour integral methods described in \cite{BKirK01} and developed in \cite{BoM-GeB-KiK-ElE96,BoM-KiK-ElE96,BoM-KiK-DoS96}, with $\kappa(r,s)$ given by
\[
\kappa(r,s) = \sum_{k} \frac{\sin \pi s}{\pi} \int_{0}^\infty u^{-2
s} \frac{d}{d u} \log \left( 1 - e^{-2 r \sqrt{\mu_k^2 + u^2}}
\right)\, d u ,
\]
where the $\mu^2_k$'s are the eigenvalues of $\Delta_Y$. From this explicit formula it is not difficult to verify the decay properties of $\kappa(r,s)$.
\end{proof}

If $\Delta_\infty = - \partial_x^2 + \Delta_Y$ on the infinite cylinder $[0,\infty) \times Y$ with Dirichlet conditions at $\{0\} \times Y$, then in \cite[Sec.\ 2]{loya04-48-1279} it was proved that $ \bz( \Delta_\infty ,s) = - \frac14 \zeta( \Delta_Y ,s)$.  Assume for the rest of this section  that $\ker \Delta_Y = \{0\}$; in particular the last term in equation \eqref{zetacyl} of Proposition \ref{prop-cylinder} vanishes. Hence, equation \eqref{zetacyl} reads: As meromorphic functions on $\C$ we have
\[
\zeta(\Delta_{N_r},s) = \frac{r}{2} \xi_Y(s) + \bz(\Delta_\infty,s) - \frac14 \zeta(\Delta_Y,s) + \kappa(r,s) .
\]
This is exactly Theorem \ref{thm-main0} for the situation at hand.

Now consider $-\partial_x^2 + \Delta_Y$ on the manifold $M_r = [-r,r] \times Y$ where we impose Dirichlet boundary conditions. Then replacing $r$ by $2r$ in equation \eqref{zetacyl} it follows that
\begin{equation} \label{zetaMr}
\zeta(\Delta_{M_r},s) = r \xi_Y(s) - \frac12 \zeta(\Delta_Y,s) + \kappa(2r,s).
\end{equation}
Recalling that $\bz( \Delta_\infty ,s) = - \frac14 \zeta( \Delta_Y ,s)$ we obtain
\[
\zeta(\Delta_{M_r},s) = r \xi_Y(s) + \bz(\Delta_\infty,s) + \bz(\Delta_\infty,s) + \kappa(2r,s),
\]
which is Theorem \ref{thm-main}  in this pure cylinder situation.

Finally, with $M_{1,r} = [-r,0] \times Y$ and $M_{2,r} = [0,r] \times Y$ and denoting by $\Delta_{1,r}$ and $\Delta_{2,r}$ the respective Dirichlet Laplacians, by Equation \eqref{zetacyl} we have
\[
\zeta(\Delta_{M_{1,r}},s) = \zeta(\Delta_{M_{2,r}},s) = \frac{r}{2} \xi_Y(s) - \frac12 \zeta(\Delta_Y,s) + \kappa(r,s) .
\]
Combining this equality with \eqref{zetaMr} we get
\[
\zeta(\Delta_{M_r}, s) - \zeta(\Delta_{M_{1,r}}, s) - \zeta(\Delta_{M_{2,r}}, s) = \frac12 \zeta(\Delta_Y,s) + \kappa(2r,s) - 2 \kappa(r,s),
\]
which implies Theorem \ref{thm-main} in this pure cylinder situation.

\subsection*{Remark}

By equation \eqref{zetacyl} of Proposition \ref{prop-cylinder}, in the case that $\ker \Delta_Y \ne \{0\}$, each of the zeta function decomposition formulas above are off by terms related to the function
\[
g(r,s) := \frac{r^{2s} \dim \ker \Delta_Y}{\pi^{2s}}\,
\zeta_R(2 s)\, ,
\]
where $\zeta_R$ is the Riemann zeta function.
The function $g(r,s)$ is ``bad'' in comparison to $\kappa(r,s)$: The function $g(r,s)$ is not an entire function of $s$ for any $r > 0$ (it has a pole at $s  = 1/2$ for all $r > 0$) and $|g(r,s)|$ does not vanish exponentially fast as $r \to \infty$ (and for $\Re s > 0$, it even increases as $r \to \infty$). This simple example explains why the main results of this paper hold only in the case $\Delta_Y$ is invertible. Throughout the rest of this paper we shall point out various details where the invertibility assumptions are crucial.

There are two possible ways to deal with the non-invertible case. The first way is to try and adapt the logarithmic surgery pseudodifferential calculus of Hassell, Mazzeo and Melrose \cite{hass98-6-255,hass95-3-115,mazz95-5-14}. However, their situation is different from ours as they do not stretch the manifold in the same way we do; they stretch it using a fixed manifold and deform the metric into a cylindrical end (or $b$-) metric. They get very precise results for the resolvent and heat kernel under the deformation and hence can get a precise understanding of the zeta function; see Section 5 of \cite{hass98-6-255}. The second way is to make further assumptions on the Laplacian. For example, one could try eigenvalue assumptions on $\Delta_{M_r}$ as was done in Park and Wojciechowski \cite{PaJ-WoK02II,PaJ-WoK05,PaJ-WoK06} or consider certain types of Laplacians such as connection Laplacians as in M{\"u}ller and  M{\"u}ller's paper \cite{MulJ-MulW06}.

\section{Trace theorems} \label{sec-trace}

In this section we study trace theorems for operators whose Schwartz kernels are continuous (but not necessarily trace-class in the functional analytic sense).

\subsection{Continuous kernels}
\label{sec-tracethm}

Let $M$ be a Riemannian manifold that is either compact with or without boundary, or a manifold with cylindrical end which means that $M$ has a decomposition
\[
M = M_0 \cup_Z \big([0,\infty)_x \times Z\big),
\]
where $M_0$ is compact with boundary $Z = \partial M_0$ and the metric $g$ on $M$ is, on the cylinder, of product type $g = dx^2 +  g_Z$ where $g_Z$ is a metric on $Z$.
We shall denote by $\mathcal C(M)$ the space of linear maps $A : L^2(M) \to L^2(M)$ with a \emph{continuous rapidly decreasing Schwartz kernel} in the sense that the Schwartz  kernel $A(z,z')$ is a continuous density on $M \times M$ that is rapidly decreasing (along the cylinders) in the case $M$ has cylindrical ends. Here, `rapidly decreasing' means the following. Let $x$ denote the variable along the cylinder and extend $x$ to be a smooth function on the rest of $M$. Then using $x$, respectively $x'$, to denote the corresponding variable on the first, respectively second, factor of $M \times M$, `rapidly decreasing' means that for any $a,b \in \N$, the density
 \[
 x^a\, (x')^b A(z,z')
 \]
on $M \times M$ is bounded. For notational convenience, throughout this paper we identify operators with their Schwartz kernels (via the Schwartz kernel theorem --- see \cite{UMelRglobal}). However, it will always be clear from context when we are in the linear map viewpoint or kernel viewpoint; note that Schwartz kernels are usually accompanied by variables such as $A(z,z')$.

Given $A \in \mathcal C(M)$, we define the trace of $A$ by integrating the density $A(z,z)$ over $M$:
\begin{equation} \label{trA}
\Tr_M(A) := \int_M A(z,z) .
\end{equation}
This gives a linear map
\[
\Tr_M : \mathcal C(M) \to \C.
\]
It is well-known that an operator $A \in \mathcal C (M)$ may not be trace-class in the functional analytic sense. For example,
Du Bois-Reymond \cite[p.\ 67]{KorT89} (cf.\ also \cite{DuiJ81}, \cite[p.\ 71]{BGoI-GoS-KrN00}) in 1876 constructed a continuous function $a: \Sph^1 \to \C$ with Fourier series
\[
a(x) \sim \sum_{k \in \Z} c_k e^{ikx},
\]
such that $\sum_{k \in \Z} |c_k| = \infty$. Define $A: L^2(\Sph^1) \to L^2(\Sph^1)$ by its Schwartz kernel
\[
A(x,y) = a(x-y)\, dy;
\]
thus as an operator, $A$ is the convolution operator
\[
Au(x) = \int_{\Sph^1} a(x - y) u(y)\, dy\quad \text{for all}\ u \in L^2(\Sph^1).
\]
For $k \in \Z$, putting $\varphi_k = \frac{1}{\sqrt{2\pi}} e^{ikx}$, observe that $\{\varphi_k\}_{k \in \Z}$ is an orthonormal basis of $L^2(\Sph^1)$ and
\[
A \varphi_k = c_k \varphi_k.
\]
Since
\[
\sum_{k \in \Z} \big| \langle A \varphi_k , \varphi_k \rangle \big| = \sum_{k \in \Z} |c_k| = \infty,
\]
it follows that $A$ is \emph{not} trace class. However, even though operators in $\mathcal C(M)$ may not be trace-class in the functional analytic sense as the simple example showed, the map
\[
\Tr_M : \mathcal C(M) \to \C
\]
has all the nice properties that the functional analytic trace does; for example, it vanishes on commutators and it is continuous with respect to any appropriate topology on continuous functions.

We shall call an operator $A: L^2(M) \to L^2(M)$ \emph{pseudo continuous}\footnote{Pseudo  continuous is supposed to be a continuous version of pseudo locality.} if for any bounded continuous functions $\varphi,\psi$ on $M$ with disjoint supports, one of which with compact support, we have
\[
\varphi A \psi \in \mathcal C(M).
\]
This is equivalent to saying that the Schwartz kernel $A(z,z')$ of $A$ is continuous and rapidly decreasing off the diagonal in $M \times M$, where `rapidly decreasing off the diagonal' only pertains to the case when $M$ has a cylindrical end; thus, the function $\varphi(z) \psi(z') A(z,z')$ is a rapidly decreasing continuous density on $M \times M$.

Let $Y \subseteq M$ be a closed codimension one submanifold of $M$ situated in the interior of $M$. Let $\gamma: C^\infty(M) \to C^\infty(Y)$ be the restriction map and let $\gamma^*$ be its adjoint, which is given by multiplying with the delta function concentrated on $Y$; that is, for $\psi \in C^\infty(Y)$, $\gamma^*(\psi)$ is the distribution on $M$ defined by
\[
 \gamma^*(\psi) := \psi \, \delta_Y
\]
where $\delta_Y$ is the delta function on $Y$.
We use the notation $\Psi^m(M)$ to denote the space of pseudodifferential operators of order $m \in \R$ on the manifold $M$ (for background on pseudodifferential operators, see e.g.\ \cite{UMelRglobal}). We say that an operator $A : L^2(M) \to L^2(M)$ is in $\Psi^m(M)$ \emph{near $Y$} if there is a function $\varphi \in C^\infty(M)$ supported near $Y$ with $\varphi \equiv 1$ near $Y$ such that
\begin{equation} \label{near}
\varphi A \varphi \in \Psi^m(M).
\end{equation}
The notion of `near $Y$' will be used in the sequel in various places.
This is equivalent to saying that on some neighborhood of $Y \times Y$ in $M \times M$, the Schwartz kernel of $A$ agrees with the Schwartz kernel of an element of $\Psi^m(M)$. In the following theorem we relate traces on $M$ to traces on $Y$.

\begin{theorem}  \label{thm-trace}
Let $m,m',m'' \in \R$ with $m,m'\leq -2$ and where at least one of $m,m',m''$ equals $-\infty$.
Let $A,B : L^2(M) \to L^2(M)$ be pseudo continuous operators and suppose that near $Y$ we have $A\in \Psi^m(M)$ and $B \in \Psi^{m'}(M)$ and let $S \in \Psi^{m''}(Y)$. Then $\gamma B A \gamma^* S \in \mathcal C (Y)$, that is,
\[
\gamma B A \gamma^* S : L^2(Y) \to L^2 (Y)
\]
and has a continuous kernel. Moreover, $A \gamma^* S \gamma B \in \mathcal C (M)$, that is,
\[
A \gamma^* S \gamma B : L^2(M) \to L^2(M)
\]
with a continuous rapidly decreasing Schwartz kernel. Furthermore,
\[
\Tr_M \big( A \gamma^* S \gamma B \big) = \Tr_Y \big( \gamma B A \gamma^* S \big).
\]
\end{theorem}
\begin{proof}
By definition of pseudo continuity (in fact, this is why this notion was introduced) and the fact that $\gamma$ and $\gamma^*$ are only relevant near $Y$,  we can reduce to the case when $A$ and $B$ are supported on a collar $(-\varepsilon, \varepsilon) \times Y$ of $Y$ where we identify $\{0\} \times Y$ with the original hypersurface $Y$. In particular, by taking $\varepsilon > 0$ sufficiently small we may assume $A \in \Psi^m(M)$ and $B \in \Psi^{m'}(M)$. By taking a partition of unity of $Y$, we can further reduce to the case when $Y$ is Euclidean space. To summarize, we may assume that $M = (-\varepsilon, \varepsilon) \times Y$ where $Y$ is Euclidean space, $A \in \Psi^m(M)$, $B \in \Psi^{m'}(M)$, and $S \in \Psi^{m''}(Y)$ all have compact support.

{\bf Step 1:} Some notations in Steps 2--4 are a little confusing so we briefly introduce the notations here in Step 1.  We denote the coordinates on $Y$ by $y$, and we denote by $(x,y)$, respectively, $(x',y')$, the coordinates on the left, respectively right, factor in $M \times M$. Consider integral operators
\[
J : L^2(Y) \to L^2(M)\ , \ K: L^2(M) \to L^2(Y)\ , \ L: L^2(Y) \to L^2(Y)
\]
with continuous Schwartz kernels, and denote their Schwartz kernels by $J(x,y,y')$, $K(y,x',y')$ and $L(y,y')$, respectively. Thus, given $\psi \in L^2(Y)$, the function $J \psi \in L^2(M)$ is given by
\[
(J \psi)(x,y) = \int J(x,y,y')\, \psi(y')\, dy',
\]
the function $L\psi \in L^2(Y)$ is given by $(L \psi)(y) = \int L(y,y')\, \psi(y')\, dy'$. Finally, given $\varphi \in L^2(M)$, the function $K\varphi \in L^2(Y)$ is given by
\[
(K \varphi)(y) = \int K(y,x',y')\, \varphi(x',y')\, dx'\, dy'.
\]
Our first observation is that if $x \in (-\varepsilon,\varepsilon)$ is fixed, then we can define an operator
\[
J(x) : L^2(Y) \to L^2(Y)
\]
in the obvious way: Given $\psi \in L^2(Y)$, $J(x) \psi \in L^2(Y)$ is the function
\[
(J(x)\psi)(y) := (J \psi)(x,y) = \int J(x,y,y')\, \psi(y')\, dy'.
\]
Similarly, if $x' \in (-\varepsilon,\varepsilon)$ is fixed, then we can define an operator
\[
K(x') : L^2(Y) \to L^2(Y)
\]
as follows: Given $\psi \in L^2(Y)$, $K(x') \psi \in L^2(Y)$ is the function
\[
(K(x')\psi)(y) := \int K(y,x',y')\, \psi(y')\, dy'.
\]
Our second observation is that we can relate composition of operators on $M = (-\varepsilon,\varepsilon) \times Y$ to composition of operators on $Y$. Consider, for example, $J$ and $L$. We have $J : L^2(Y) \to L^2(M)$ and $L: L^2(Y) \to L^2(Y)$, so
\[
J \circ L : L^2(Y) \to L^2(M).
\]
We claim that the Schwartz kernel of this operator is
\begin{equation} \label{JL}
(J \circ L)(x,y,y') = (J(x) \circ_Y L)(y,y'),
\end{equation}
where the subscript $Y$ in $J (x) \circ_Y L$ refers to the composition of the operators $J(x): L^2(Y) \to L^2(Y)$ and $L : L^2(Y) \to L^2(Y)$ as operators on $Y$. To prove \eqref{JL} we simply compute: Given $\psi \in L^2(Y)$, $(J \circ L)\psi \in L^2(M)$ is the function
\begin{align*}
(J(L\psi))(x,y) = \int J(x,y,z) \, (L\psi)(z)\, dz &
= \int J(x,y,z) \, \Big( \int L(z,y')\, \psi(y')  dy' \Big) dz \\
&
= \int \Big( \int J(x,y,z) \, L(z,y') \, dz \Big) \psi(y') dy'  .
\end{align*}
It follows that the Schwartz kernel of $J \circ L$ is
\[
(J \circ L)(x,y,y') = \int J(x,y,z) \, L(z,y') \, dz ,
\]
which is exactly $(J(x) \circ_Y L)(y,y')$ as can be readily checked. Similarly, we have the following formulas: $J \circ K : L^2(M) \to L^2(M)$ has the Schwartz kernel
\begin{equation} \label{JK}
  (J \circ K)(x,y,x',y') = (J(x) \circ_Y K(x'))(y,y'),
\end{equation}
and finally, $K \circ J : L^2(Y) \to L^2(Y)$ has the Schwartz kernel
\begin{equation} \label{KJ}
 (K \circ J)(y,y') = \int (K(x) \circ_Y J(x))(y,y')\, dx.
\end{equation}

{\bf Step 2:} We now consider $A$. The Schwartz kernel of $A$ is of the form (dropping density factors for simplicity)
\[
A(x,y,x',y') = \int e^{i(x - x') \xi + i(y - y') \cdot \eta} \, a(x,y,\xi,\eta) \dbar \xi \dbar \eta
\]
where $a(x,y,\xi,\eta)$ is a symbol in $(\xi,\eta)$ (the dual variables to $(x,y)$) of order $m$ and $\dbar \xi = d \xi/2\pi$ and $\dbar \eta = d \eta/(2\pi)^{\dim Y}$. Since $\{0\} \times Y$ is the original hypersurface in $M$,  $\gamma^*$ is multiplication by the delta function at $x = 0$, so it follows that the Schwartz kernel of $A \gamma^*$ is
\[
A \gamma^* (x,y,y') = \int e^{i x \xi + i(y - y') \cdot \eta} \, a(x,y,\xi,\eta) \dbar \xi \dbar \eta.
\]
By assumption, $m \leq -2$, so the integral
\[
\alpha(x,y,\eta) : = \int_\R e^{ix \xi} \, a(x,y,\xi,\eta) \dbar \xi
\]
is absolutely convergent, and moreover it is easy to check that  $\alpha(x,y,\eta)$ is a symbol of order $m$ in $\eta$ that is smooth in $y$ and continuous in $x$ (it may not be smooth in $x$ unless $m = -\infty$, but all we need is that it is continuous in $x$), and
\begin{equation*} \label{Agammas}
A \gamma^* (x,y,y') = \int e^{i(y - y') \cdot \eta} \, \alpha(x,y,\eta) \dbar \eta .
\end{equation*}
By the properties of $\alpha$ it follows that for fixed $x \in (-\varepsilon, \varepsilon)$, in terms of the variables $(y,y')$, $A \gamma^* (x,y,y')$ is the Schwartz kernel of an element of $\Psi^m(Y)$. We denote this element by $A \gamma^*(x)$. It's clear that $x \mapsto A \gamma^*(x) \in \Psi^m(Y)$ is continuous.

Now, since $S \in \Psi^{m''}(Y)$, by \eqref{JL} with $J = A \gamma^*$ and $L = S$, it follows that
\[
(A \gamma^* S)(x,y,y') = ((A \gamma^*) (x) \circ_Y S)(y,y').
\]
We remark that technically speaking, the derivation of \eqref{JL} used the fact that $J$ and $L$ had continuous Schwartz kernels, and $A \gamma^*$ and $S$ may not have continuous Schwartz kernels (unless $m$ and $m''$ were sufficiently negative); however, we can still apply \eqref{JL} by the standard continuity arguments, see Chapter 2 of  \cite{UMelRglobal}.
Since $A \gamma^*(x) \in \Psi^m(Y)$ and $S \in \Psi^{m''}(Y)$, it follows that
\[
(A \gamma^* S)(x) := A \gamma^*(x) \circ_Y S \in \Psi^{m + m''}(Y)
\]
is a family of pseudodifferential operators on $Y$ of order $m + m''$ depending continuously on $x$.

{\bf Step 3:} Now let us consider $\gamma B$. On the collar, the Schwartz kernel of $B$ is of the form
\[
B(x,y,x',y') = \int e^{i(x - x') \xi + i(y - y') \cdot \eta} \, b(x,y,\xi,\eta) \dbar \xi \dbar \eta ,
\]
where $b(x,y,\xi,\eta)$ is a symbol in $(\xi,\eta)$ of order $m'$. Thus, recalling that $\gamma$ is restriction to $\{0\} \times Y$, the Schwartz kernel of $\gamma B$ is
\[
\gamma B(y,x',y') = \int e^{-i x' \xi + i(y - y') \cdot \eta} \, b(0,y,\xi,\eta) \dbar \xi \dbar \eta.
\]
Recalling that $m' \leq -2$, it follows that the integral
\[
\beta(x',y,\eta) : = \int_\R e^{-ix'\xi} \, b(0,y,\xi,\eta) \dbar \xi
\]
is absolutely convergent and defines a symbol of order $m'$ in $\eta$ that is continuous in $x'$ and smooth in $y$, and
\[
(\gamma B)(y,x',y') = \int e^{i(y - y') \cdot \eta} \, \beta(x',y, \eta) \dbar \eta.
\]
Directly from this formula we observe that for fixed $x' \in (-\varepsilon, \varepsilon)$, in terms of the variables $(y,y')$ we have $\gamma B(x') \in \Psi^{m'}(Y)$ and moreover, $x' \mapsto \gamma B(x') \in \Psi^{m'}(Y)$ is continuous.

{\bf Step 4:} We now put Steps 1--3 together to prove our result. First,  by \eqref{JK} with $J = A \gamma^* S$ and $K = \gamma B$, observe that the Schwartz kernel of $A \gamma^* S \gamma B = A \gamma^* S \circ B$ is given by
\begin{equation} \label{Agamma*}
(A \gamma^* S \gamma B)(x,y,x',y') = ((A \gamma^* S)(x) \circ_Y (\gamma B)(x')) (y,y').
\end{equation}
Since $(A \gamma^* S)(x) \in \Psi^{m + m''}(Y)$ and $(\gamma B)(x') \in \Psi^{m'}(Y)$ depend continuously on $x$ and $x'$, respectively, it follows that
\[
(x,x') \mapsto (A \gamma^* S)(x) \circ_Y (\gamma B)(x') \in \Psi^{m + m' + m''}(Y)
\]
is a continuous map. By assumption, at least one of $m,m',m''$ is $-\infty$, so this continuous map is a map into the smoothing operators on $Y$. Therefore, the Schwartz kernel $A \gamma^* S \gamma B(x,y,x',y')$ is smooth in $(y,y')$ (and continuous in $(x,x')$); in particular, we have a map
\[
A \gamma^* S \gamma B : L^2(M) \to L^2(M)
\]
with a continuous Schwartz kernel. Second, in view of \eqref{KJ} with $K = \gamma B$ and $J = A \gamma^* S$, observe that the Schwartz kernel of $\gamma B A \gamma^* S = \gamma B \circ A \gamma^* S$ is given by
\begin{equation} \label{gammaB2}
(\gamma B A \gamma^* S)(y,y') = \int ((\gamma B) (x) \circ_Y (A \gamma^* S)(x)) (y,y')\, dx.
\end{equation}
Since $(\gamma B)(x) \in \Psi^{m'}(Y)$ and $(A \gamma^* S)(x) \in \Psi^{m + m''}(Y)$,
\[
x \mapsto (\gamma B) (x) \circ_Y (A \gamma^* S)(x) \in \Psi^{m + m' + m''}(Y)
\]
is a continuous map into the smoothing operators, since at least one of $m$, $ m'$, $ m''$ is $-\infty$. Therefore,  $\gamma B A \gamma^* S$ is in fact a smoothing operator on $Y$ and hence in particular has a continuous Schwartz kernel. To see the trace property, note that by \eqref{Agamma*}, we have
\begin{align}
\Tr_M (A \gamma^* S \gamma B) & = \int \int_Y (A \gamma^* S \gamma B)(x,y,x,y) \, dy\, dx \nonumber \\ \nonumber &= \int \Big( \int ((A \gamma^* S)(x) \circ_Y (\gamma B)(x))(y,y)\, dy \Big) dx \\ \label{trace1} &= \int \Tr_Y ((A \gamma^* S)(x) \circ_Y (\gamma B)(x)) \, dx.
\end{align}
On the other hand, by \eqref{gammaB2}, we have
\begin{align}
\nonumber \Tr_Y (\gamma B A \gamma^* S) & = \int (\gamma B A \gamma^* S) (y,y)\, dy\\ \nonumber & = \int \Big( \int ((\gamma B) (x) \circ_Y (A \gamma^* S)(x)) (y,y)\, dx\Big) dy\\ \nonumber & = \int \Big( \int ((\gamma B) (x) \circ_Y (A \gamma^* S)(x)) (y,y)\, dy\Big) dx\\
\label{trace2} & = \int \Tr_Y((\gamma B) (x) \circ_Y (A \gamma^* S)(x)) \, dx .
\end{align}
Since for each $x \in (-\varepsilon , \varepsilon)$, $(\gamma B)(x)$ and $(A \gamma^* S)(x)$ are pseudodifferential operators on $Y$, one of which is of order $-\infty$, it is well known that (see \cite[Ch.\ 3]{UMelRglobal})
\[
\Tr_Y((\gamma B) (x) \circ_Y (A \gamma^* S)(x)) = \Tr_Y ((A \gamma^* S)(x) \circ_Y (\gamma B)(x)).
\]
Hence, \eqref{trace1} and \eqref{trace2} are identical, so $\Tr_M (A \gamma^* S \gamma B) =
\Tr_Y (\gamma B A \gamma^* S)$, and our proof is complete.
\end{proof}

\subsection{A relative trace theorem}
\label{sec-R}

Using Theorem \ref{thm-trace} we derive the following result that will be used to prove a gluing formula (Theorem \ref{thm-Cas}). Recall that $M$ is either a compact manifold (with or without boundary) or a manifold with cylindrical end and $Y$ is a closed codimension one submanifold in the interior of $M$.

\begin{theorem} \label{thm-reltrace} Let $A_1,B_1,A_2,B_2 : L^2(M) \to L^2(M)$ be pseudo continuous such that near  $Y$ they define operators in $\Psi^{-2}(M)$ and such that, near $Y$, $A_1 - B_1 \in \Psi^{-\infty}(M)$ and $A_2 - B_2 \in \Psi^{-\infty}(M)$,  and let $S,T \in \Psi^*(Y)$ such that $S - T \in \Psi^{-\infty}(Y)$. Then we have $L^2$ maps
\[
A_1 \gamma^* S \gamma A_2 - B_1 \gamma^* T \gamma B_2 : L^2(M) \to L^2(M)
\]
and
\[
\gamma A_2 A_1 \gamma^* S - \gamma B_2 B_1 \gamma^* T: L^2(Y) \to L^2(Y)
\]
that have continuous Schwartz kernels. Moreover,
\[
\Tr_M \big( A_1 \gamma^* S \gamma A_2 - B_1 \gamma^* T \gamma B_2 \big) = \Tr_Y \big( \gamma A_2 A_1 \gamma^* S - \gamma B_2 B_1 \gamma^* T \big).
\]
\end{theorem}
\begin{proof} Observe that
\begin{multline*}
A_1 \gamma^* S \gamma A_2 - B_1 \gamma^* T \gamma B_2 = \\
(A_1 - B_1) \gamma^* S \gamma A_2 \ \, +\ \,  B_1 \gamma^* (S - T) \gamma A_2\  \, +\ \, B_1 \gamma^* T \gamma (A_2 - B_2),
\end{multline*}
and
\begin{multline*}
\gamma A_2 A_1 \gamma^* S - \gamma B_2 B_1 \gamma^* T = \\
\gamma A_2(A_1 - B_1) \gamma^* S  \ \, +\ \, \gamma  A_2 B_1 \gamma^* (S - T) \  \, +\ \, \gamma (A_2 - B_2) B_1 \gamma^* T,
\end{multline*}
so the Schwartz kernel properties follow from Theorem \ref{thm-trace} since all the difference operators $(A_i - B_i)$ and $S - T$ are smoothing near $Y$. Also by Theorem \ref{thm-trace}, we have
\begin{align*}
\Tr_M &  \big(A_1 \gamma^* S \gamma A_2 - B_1 \gamma^* T \gamma B_2 \big)
\\ & = \Tr_M \big( (A_1 - B_1) \gamma^* S \gamma A_2 \ \, +\ \,  B_1 \gamma^* (S - T) \gamma A_2\  \, +\ \, B_1 \gamma^* T \gamma (A_2 - B_2) \big)\\
& = \Tr_M\big( (A_1 - B_1) \gamma^* S \gamma A_2\big) \ \, +\ \,  \Tr_M\big(B_1 \gamma^* (S - T) \gamma  A_2\big)\\ & \hspace{6.5cm}  \, +\ \, \Tr_M\big( B_1 \gamma^* T \gamma (A_2 - B_2) \big)\\
& = \Tr_Y\big( \gamma A_2(A_1 - B_1) \gamma^* S \big) \ \, +\ \,  \Tr_Y\big(\gamma  A_2 B_1 \gamma^* (S - T) \big)\\ & \hspace{6.5cm}  \, +\ \, \Tr_Y\big(\gamma (A_2 - B_2) B_1 \gamma^* T  \big)\\
& = \Tr_Y\big( \gamma A_2(A_1 - B_1) \gamma^* S\ \,  +\ \, \gamma  A_2 B_1 \gamma^* (S - T) \ \, +\ \, \gamma (A_2 - B_2) B_1 \gamma^* T  \big)\\
& = \Tr_Y\big( \gamma A_2 A_1 \gamma^* S - \gamma B_2 B_1 \gamma^* T  \big).
\end{align*}
\end{proof}

\section{A gluing formula for the zeta function} \label{sec-glueform}

Let $M$ be a compact manifold with or without boundary decomposed as a union $M = M_- \cup_Y M_+$ where $Y \subseteq M$ is a codimension one submanifold in the interior of $M$, $M_\pm$ are smooth manifolds with boundary such that $Y \subseteq \partial M_\pm$ and $Y = M_-  \cap M_+$; see Figure \ref{fig-Mplusminus}.  Following \cite{loya04-48-1279} we give a formula for the zeta function on $M$ in terms of the zeta functions on $M_-$ and $M_+$ and Dirichlet-to-Neumann maps.

\begin{figure}[h!] \centering
\includegraphics{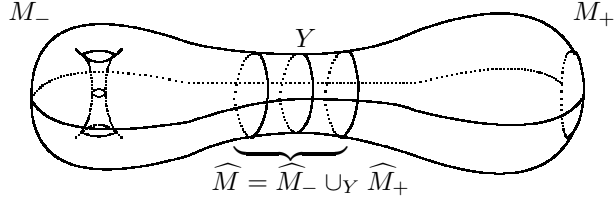} \caption{Here, $M$ is a manifold with boundary (the boundary is at the far right) that is partitioned into submanifolds $M_-$ (to the left of $Y$) and $M_+$ (to the right of $Y$). The submanifold $\widehat{M}$ has a similar partition.} \label{fig-Mplusminus}
\end{figure}

\subsection{Dirichlet-to-Neumann maps} \label{subsec-DN}

Let $\Delta$ be a Laplace-type operator on $M$; we do not assume that $\Delta$ is of product-type near $Y$.  At $\partial M$ we \emph{always} impose the Dirichlet boundary condition if in fact $M$ has a boundary. The Dirichlet-to-Neumann maps for $M_\pm$ are described as follows. We denote by $\Delta_\pm$ the restrictions of $\Delta$ to $M_\pm$ with Dirichlet boundary conditions at $Y$ (and also at $\partial M \cap M_\pm$). Thus, in Figure \ref{fig-Mplusminus}, $\Delta_+$ has Dirichlet conditions both at $Y$ and at the boundary of $M$ on the far right and $\Delta_-$ just has Dirichlet conditions at $Y$.
For simplicity we henceforth shall use the notation $\Delta(\lambda) := \Delta - \lambda$ with similar notations with $\Delta$ replaced by $\Delta_{\pm}$ or with any Laplace-type operator. Consider $M_+$ (just so that we do not have to use the notation $\pm$).
For $\lambda \in \C \setminus \mathrm{spec}(\Delta_+)$ and $\psi \in C^\infty(Y)$ we claim  there is a unique solution $\phi$ on $M_+$ to the boundary value problem
\begin{equation} \label{DP}
\text{$\Delta(\lambda) \phi = 0$ on $M_+$ and $\phi|_Y = \psi$,}
\end{equation}
and where $\phi = 0$ on $\partial M \cap M_+$.
Indeed, with $\widetilde{\psi}$ denoting any smooth extension of $\psi$ to $M_+$ vanishing at $\partial M \cap M_+$, it is easy to check that
\[
\phi := \widetilde{\psi} - \Delta_+(\lambda)^{-1} \Delta(\lambda) \widetilde{\psi}
\]
satisfies $\phi|_Y = \psi$ and $\Delta(\lambda) \phi = 0$. Here, we recall that $\Delta_+(\lambda)^{-1} = (\Delta_+ - \lambda)^{-1}$ is the resolvent of the Dirichlet Laplacian on $M_+$. This proves existence, and uniqueness follows from the fact that if $\phi'$ has the same properties as $\phi$, then $\phi - \phi'$ vanishes on $Y$ and $\partial M \cap M_+$, and therefore $\phi - \phi'$ is in the domain of $\Delta_+(\lambda)$. Moreover, $\Delta_+(\lambda) (\phi - \phi') = 0$, which implies that $\phi - \phi' = 0$ because $\lambda \notin \mathrm{spec}(\Delta_+)$. Now with $\phi$ satisfying \eqref{DP} we define
\[
\mathcal{N}_+(\lambda) \psi := \frac{\partial}{\partial \vec{n}} \, \phi \Big|_{Y} ,
\]
where $\vec{n}$ denotes the outer unit vector field on $Y$ for $M_+$. This defines a map
\[
\mathcal{N}_+(\lambda) : C^\infty(Y) \to C^\infty(Y),
\]
called the \emph{Dirichlet-to-Neumann map} for $M_{+}$.
For $\lambda \in \C \setminus \big(\mathrm{spec}(\Delta_+) \cup \mathrm{spec}(\Delta_-) \big)$, we denote by
\[
R(\lambda) : C^\infty(Y) \to C^\infty(Y),
\]
the sum of the Dirichlet-to-Neumann maps:
\begin{equation} \label{defR}
R(\lambda) := \mathcal{N}_-(\lambda) + \mathcal{N}_+(\lambda) ,
\end{equation}
where $\mathcal{N}_-(\lambda)$ is defined using the outer unit vector field on $Y$ for $M_-$ (the vector $-\vec{n}$ where $\vec{n}$ was the outer unit vector field on $Y$ for $M_+$).
We shall call $R(\lambda)$ the \emph{Dirichlet-to-Neumann map} for the partitioned manifold $M = M_- \cup_Y M_+$. By construction, $R(\lambda)$ is an analytic function of $\lambda \in \C \setminus \big(\mathrm{spec}(\Delta_+) \cup \mathrm{spec}(\Delta_-) \big)$.
One can show that $R(\lambda)$ is invertible if and only if $\lambda \notin \mathrm{spec}(\Delta)$ (cf.\ \cite[Th.\ 2.1]{CarG02}), in which case
\[
R(\lambda)^{-1} = \gamma \Delta(\lambda)^{-1} \gamma^* .
\]

The operator $R(\lambda)$ depends on $\lambda$ in a special way, described as follows.
First we define the residue space. Let $\Psi^{-\infty}(Y)$ denote the space of smoothing operators on $Y$. Identifying operators with their Schwartz kernels we consider
\[
\Psi^{-\infty}(Y) \equiv C^\infty(Y \times Y, \Omega_R),
\]
where $\Omega_R$ is the bundle of densities over $Y$ lifted to $Y \times Y$ on the right. With this identification, the space of smoothing operators $\Psi^{-\infty}(Y)$ inherits a natural Fr\'echet topology. We shall call a subset $\Lambda \subseteq \C$ \emph{sectorial} if outside some neighborhood of the origin, $\Lambda$ equals a sector (solid angle). Given a sectorial subset $\Lambda \subseteq \C$, we define $\Psi^{-\infty}_\Lambda(Y)$ as the space of $\Psi^{-\infty}(Y)$-valued Schwartz functions on $\Lambda$,
\[
\Psi^{-\infty}_\Lambda(Y) = \mathcal{S}(\Lambda, \Psi^{-\infty}(Y)).
\]
(The Schwartz functions on $\Lambda$ with values in any Fr\'echet space is well-defined.)
Now we say that a parameter dependent operator $A(\lambda)$, $\lambda \in \Lambda$, has \emph{weight $m \in \R$} if the following conditions are satisfied: $1)$ off a neighborhood of the diagonal in $Y \times Y$, the Schwartz kernel of $A(\lambda)$ is given by a Schwartz kernel of an element of $\Psi^{-\infty}_\Lambda(Y)$. $2)$ Locally on a coordinate patch on $Y$, the Schwartz kernel of $A(\lambda)$ is of the form
\[
A(\lambda,y,y') = \int e^{i(y - y') \cdot \eta} \, a(\lambda,y,\eta) \dbar \eta ,
\]
where $a(\lambda,y,\eta)$ satisfies the estimates: Given $\alpha,\beta,\gamma$ there is a constant $C$ such that
\begin{equation} \label{asymbol}
\big|\partial_\lambda^\alpha \partial_y^\beta \partial_\eta^\gamma a(\lambda,y,\eta)\big|
\leq C (1 + |\lambda| + |\eta|)^{m- |\alpha| - |\gamma|}.
\end{equation}
The following result is not difficult to verify; see \cite{BuD-FrL-KaT92}.
\begin{proposition} \label{prop-R}
For any sectorial $\Lambda \subseteq \C$ such that $R(\lambda)$ is defined for all $\lambda \in \Lambda$, the operator $R(\lambda)$ is analytic in $\lambda$ and parameter dependent of weight $1$.
\end{proposition}

Indeed, one proof follows by examining the formula $R(\lambda)^{-1} = \gamma \Delta(\lambda)^{-1} \gamma^*$ in local coordinates and using well-known pseudodifferential facts about $\Delta(\lambda)^{-1}$; this will show that $R(\lambda)^{-1}$ is of weight $-1$, which implies that $R(\lambda)$ is of weight $1$.

Let $\widehat{M} \subseteq M$ be a smooth manifold with boundary containing $Y$ in its interior and suppose that $\widehat{M} = \widehat{M}_- \cup_Y \widehat{M}_+$ where $\widehat{M}_-$ and $\widehat{M}_+$ are manifolds with boundary such that $Y \subseteq \partial \widehat{M}_\pm$ and $\widehat{M}_- \cap \widehat{M}_+ = Y$; see Figure \ref{fig-Mplusminus}.   Later in our proof of the analytic surgery formulas, $\widehat{M}$ will just be a collar $[-1,1] \times Y$ of $Y$ with $\widehat{M}_- = [-1,0]\times Y$ and $\widehat{M}_+ = [0,1] \times Y$.

We now play the same game with $\widehat{M}$ as we did with $M$. Thus, let $\widehat{\Delta}$ be the restriction of $\Delta$ to $\widehat{M}$ where we put Dirichlet boundary conditions at $\partial \widehat{M}$ and let $\widehat{\Delta}_\pm$ be the restriction of $\widehat{\Delta}$ to $\widehat{M}_\pm$ with additional Dirichlet boundary conditions at $Y$.
For $\lambda \in \C \setminus \big(\mathrm{spec}(\widehat{\Delta}_+) \cup \mathrm{spec}(\widehat{\Delta}_-) \big)$, let
\[
\widehat{R}(\lambda) : C^\infty(Y) \to C^\infty(Y)
\]
be the Dirichlet-to-Neumann map for $\widehat{M} = \widehat{M}_- \cup_Y \widehat{M}_+$. As before, $\widehat{R}(\lambda)$ is invertible if and only if $\lambda \notin \mathrm{spec}(\widehat{\Delta})$, in which case
\[
\widehat{R}(\lambda)^{-1} = \gamma \widehat{\Delta}(\lambda)^{-1} \gamma^*.
\]

The following proposition can be seen from the proof of Proposition 4.1 in \cite{loya04-48-1279}.

\begin{proposition}\label{prop-RR0}
For any sectorial $\Lambda \subseteq \C$ not overlapping the spectra of $\Delta_\pm$, $\widehat{\Delta}_\pm$,
\[
R(\lambda) - \widehat{R}(\lambda) \in \Psi^{-\infty}_\Lambda(Y) .
\]
\end{proposition}

This proposition is in some sense ``obvious'' because $\Delta(\lambda)^{-1}$ and $\widehat{\Delta}(\lambda)^{-1}$ have the identical symbolic structure near $Y$ and it follows that $\Delta(\lambda)^{-1} - \widehat{\Delta}(\lambda)^{-1} \in \Psi^{-\infty}_\Lambda(M)$ near $Y$ (see the discussion around \eqref{near} for the notion of `near $Y$'). Applying $\gamma$ and $\gamma^*$ to both sides of $\Delta(\lambda)^{-1} - \widehat{\Delta}(\lambda)^{-1}$ implies that $R(\lambda)^{-1} - \widehat{R}(\lambda)^{-1}$ is in $\Psi^{-\infty}_\Lambda(Y)$ and then using that
\[
R(\lambda) - \widehat{R}(\lambda) = R(\lambda) \big( \widehat{R}(\lambda)^{-1} - R(\lambda)^{-1} \big) \widehat{R}(\lambda),
\]
and the fact that $\Psi^{-\infty}_\Lambda(Y)$ is an ideal within the space of all parameter dependent operators of any weight (cf.\ Lemma \ref{lem-ideal} for the proof of a related result) proves the proposition.

\subsection{A zeta function gluing formula for the compact case}
\label{sec-var}

We now prove a `relative gluing formula' for zeta functions. For a similar result for the zeta determinant, see Proposition 4.4 of \cite{loya04-48-1279}. The notation $\Delta, \Delta_\pm,\ldots$ in the following theorem are described in Section \ref{subsec-DN}.

\begin{theorem} \label{thm-Cas}
As meromorphic functions on $\mathbb{C}$ we have
\begin{multline*}
\zeta (\Delta_-, s ) + \zeta(\Delta_+,s) - \zeta ( \Delta,s )     \\ = f(s) + \frac{i}{2\pi} \int_\Gamma \lambda^{-s} \Tr_Y \big[ R(\lambda)^{-1} R'(\lambda) - \widehat{R}(\lambda)^{-1} \widehat{R}'(\lambda) \big] d\lambda ,
\end{multline*}
where
\[
f(s) = \zeta(\widehat{\Delta}_{-},s) + \zeta(\widehat{\Delta}_{+},s) - \zeta(\widehat{\Delta},s),
\]
$\Gamma = \{\lambda \in \C\, ;\, \Re \lambda = c\}$ with $c > 0$ sufficiently small, and where recall that $R(\lambda)$ and $\widehat{R}(\lambda)$ denote, respectively, the Dirichlet-to-Neumann maps for the partitioned manifolds $M = M_- \cup_Y M_+$ and $\widehat{M} = \widehat{M}_- \cup_Y \widehat{M}_+$. (The primes in $R'(\lambda)$ and $\widehat{R}'(\lambda)$ denote differentiation with respect to $\lambda$.)
\end{theorem}
\begin{proof}
Let $\Delta_D = \Delta_- \oplus \Delta_+$,
which is the Dirichlet Laplacian on the disjoint union $M_- \sqcup M_+$ and let $\widehat{\Delta}_D = \widehat{\Delta}_- \oplus \widehat{\Delta}_+$, the Dirichlet Laplacian on $\widehat{M}_- \sqcup \widehat{M}_+$. Then our theorem is the equality
\begin{multline} \label{Cas}
\zeta (\Delta_D,s) - \zeta ( \Delta,s ) - \zeta(\widehat{\Delta}_D,s) + \zeta(\widehat{\Delta},s)    \\ = \frac{i}{2\pi} \int_\Gamma \lambda^{-s} \Tr_Y \big[ R(\lambda)^{-1} R'(\lambda) - \widehat{R}(\lambda)^{-1} \widehat{R}'(\lambda) \big] d\lambda .
\end{multline}
We break up the proof of this equality in three steps.

{\bf Step 1:}
We claim that
\begin{equation} \label{DeltaD-1}
\Delta_D(\lambda)^{-1} = \Delta(\lambda)^{-1} - \Delta(\lambda)^{-1} \gamma^* R(\lambda) \gamma \Delta(\lambda)^{-1} ,
\end{equation}
where we use $\Delta_D(\lambda)$ and $\Delta(\lambda)$ to denote $\Delta_D - \lambda$ and $\Delta -\lambda$, respectively, and where we assume all operators in \eqref{DeltaD-1} are defined at $\lambda$.
To prove \eqref{DeltaD-1}, let $A(\lambda)$ denote the operator on the right-hand side of \eqref{DeltaD-1}, let $f \in C^\infty(X)$ where $X = M_- \sqcup M_+$, and define
\[
u:= A(\lambda) f \in C^\infty(X).
\]
Then it follows that
\[
\Delta(\lambda) u = f \quad \text{in the interior of $X$} \quad \mbox{and}\quad u|_Y = 0.
\]
Indeed, the first condition is obvious and the second condition is just a computation:
\begin{align*}
u|_Y = \gamma A(\lambda) f & = \gamma \Delta(\lambda)^{-1} f - \gamma \Delta(\lambda)^{-1} \gamma^* R(\lambda) \gamma \Delta(\lambda)^{-1}f \\ & = \gamma \Delta(\lambda)^{-1}f - R(\lambda)^{-1} R(\lambda) \gamma \Delta(\lambda)^{-1}f \\ & = \gamma \Delta(\lambda)^{-1}f - \gamma \Delta(\lambda)^{-1}f = 0.
\end{align*}
Thus, $A(\lambda)$ is indeed the resolvent $\Delta_D(\lambda)^{-1}$ of the Dirichlet Laplacian $\Delta_D(\lambda)$.

Similarly, with
\[
\widehat{\Delta}_D = \widehat{\Delta}_- \oplus \widehat{\Delta}_+,
\]
we have
\[
\widehat{\Delta}_D(\lambda)^{-1} = \widehat{\Delta}(\lambda)^{-1} - \widehat{\Delta}(\lambda)^{-1} \gamma^* \widehat{R}(\lambda) \gamma \widehat{\Delta}(\lambda)^{-1} .
\]

{\bf Step 2:} We now prove that
\begin{multline*}
\Tr_Y \big[ R(\lambda)^{-1}R'(\lambda) - \widehat{R}(\lambda)^{-1} \widehat{R}'(\lambda) \big] \\
= \Tr_M \big[ \Delta_D(\lambda)^{-1} - \Delta(\lambda)^{-1} - \widehat{\Delta}_D(\lambda)^{-1} + \widehat{\Delta}(\lambda)^{-1} \big]
\end{multline*}
and in the process we shall verify that these traces are actually defined. In fact,
$R(\lambda)^{-1}R'(\lambda) - \widehat{R}(\lambda)^{-1} \widehat{R}'(\lambda)$ is smoothing because of Proposition \ref{prop-RR0}, so  we shall consider first the left-hand side. Using that
\[
\frac{d}{d \lambda} \Big( R(\lambda)^{-1} \Big)= - R(\lambda)^{-1} \, R'(\lambda)\, R(\lambda)^{-1}\ \ \Longrightarrow\ \ R(\lambda)^{-1} R'(\lambda) = - \frac{d}{d \lambda} \Big( R(\lambda)^{-1} \Big)
R(\lambda),
\]
with a similar formula for $\frac{d}{d \lambda} \widehat{R}(\lambda)^{-1}$, we see that
\begin{multline*}
\Tr_Y \big[ R(\lambda)^{-1}R'(\lambda) - \widehat{R}(\lambda)^{-1} \widehat{R}'(\lambda) \big]\\
= - \Tr_Y \Big[  \frac{d}{d \lambda} \Big( R(\lambda)^{-1} \Big) R (\lambda) -  \frac{d}{d \lambda} \Big( \widehat{R}(\lambda)^{-1} \Big) \widehat{R}(\lambda) \Big] .
\end{multline*}
Since $R(\lambda)^{-1} = \gamma \Delta(\lambda)^{-1} \gamma^* = \gamma (\Delta - \lambda)^{-1} \gamma^*$, we have
\[
\frac{d}{d\lambda} \Big( R(\lambda)^{-1} \Big) = \gamma \Delta(\lambda)^{-2} \gamma^* .
\]
Thus,
\begin{multline*}
\Tr_Y \big[ R(\lambda)^{-1}R'(\lambda) - \widehat{R}(\lambda)^{-1} \widehat{R}'(\lambda) \big] \\
= \Tr_Y \Big[ \gamma \widehat{\Delta} (\lambda)^{-1} \widehat{\Delta}(\lambda)^{-1} \gamma^* \widehat{R}(\lambda) - \gamma \Delta(\lambda)^{-1} \Delta(\lambda)^{-1} \gamma^* R (\lambda) \Big] .
\end{multline*}
We claim that $A_1 = A_2 = \Delta(\lambda)^{-1}$ and $B_1 = B_2 = \widehat{\Delta}(\lambda)^{-1}$ satisfy the hypotheses of Theorem \ref{thm-reltrace}, where we extend the Schwartz kernel of $\widehat{\Delta}(\lambda)^{-1}$ to $M \times M$ by defining it to be zero off of $\widehat{M} \times \widehat{M}$. Indeed, by properties of pseudodifferential operators, both $\Delta(\lambda)^{-1}$ and $\widehat{\Delta}(\lambda)^{-1}$ are pseudo continuous on $M$ and near $Y$ they define operators in $\Psi^{-2}(M)$ (of course, $\Delta(\lambda)^{-1}$ is an operator in $\Psi^{-2}(M)$ globally). Also recall that $\Delta(\lambda)^{-1} - \widehat{\Delta}(\lambda)^{-1}$ is smoothing near $Y$ and $R(\lambda) - \widehat{R}(\lambda)$ is smoothing (see Proposition \ref{prop-RR0} and its discussion). Thus, the conditions of Theorem \ref{thm-reltrace} are satisfied so
\[
K(\lambda) :=  \widehat{\Delta}(\lambda)^{-1} \gamma^* \widehat{R}(\lambda) \gamma \widehat{\Delta} (\lambda)^{-1} - \Delta(\lambda)^{-1} \gamma^* R (\lambda) \gamma \Delta(\lambda)^{-1}
\]
has a continuous Schwartz kernel on $M \times M$, and
\begin{equation} \label{TYK}
\Tr_Y \big[ R(\lambda)^{-1}R'(\lambda) - \widehat{R}(\lambda)^{-1} \widehat{R}'(\lambda) \big]  = \Tr_M K(\lambda) = \int_M K(\lambda)|_{\mathrm{Diag}}dg,
\end{equation}
where $\mathrm{Diag}$ is the diagonal in $M \times M$ and $dg$ is the Riemannian density.
Recalling the formulas for $\Delta_D(\lambda)^{-1}$ and $\widehat{\Delta}_D(\lambda)^{-1}$ in Step 1 shows that
\[
K(\lambda) = \Delta_D(\lambda)^{-1} - \Delta(\lambda)^{-1} - \widehat{\Delta}_D(\lambda)^{-1} + \widehat{\Delta}(\lambda)^{-1},
\]
which completes the proof of Step 2.
\begin{figure} \centering
\includegraphics{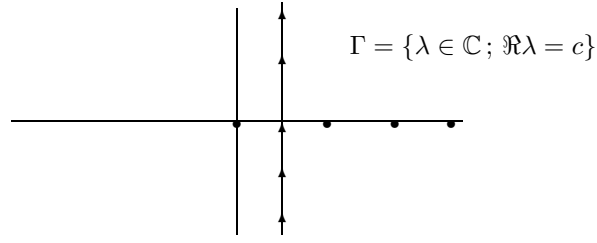} \caption{The dots are the points where $R(\lambda)^{-1}$ and $\widehat{R}(\lambda)^{-1}$ are not defined (the union of the spectra of $\Delta_D$, $\Delta$, $\widehat{\Delta}_D$ and $\widehat{\Delta}$). The constant $c > 0$ is any positive real number such that $R(\lambda)^{-1}$ and $\widehat{R}(\lambda)^{-1}$ are defined for $\lambda \in (0,c]$.} \label{fig-Gamma}
\end{figure}

{\bf Step 3:} We can now prove our result. Let $\Gamma = \{\lambda \in \C\, ;\, \Re \lambda = c\}$ be as in Figure \ref{fig-Gamma}. Multiplying both sides of the equation \eqref{TYK} by $\frac{i}{2\pi} \lambda^{-s}$ and then integrating over $\Gamma$, for any $s \in \C$ we have
\[
\frac{i}{2\pi} \int_\Gamma \lambda^{-s} \Tr_Y \big[ R(\lambda)^{-1}R'(\lambda) - \widehat{R}(\lambda)^{-1} \widehat{R}'(\lambda) \big] d\lambda
= \frac{i}{2\pi} \int_\Gamma \int_M \lambda^{-s} K(\lambda)|_{\mathrm{Diag}} dg\, d\lambda;
\]
by analyticity these integrals are defined independent of $c > 0$ chosen as in Figure \ref{fig-Gamma}. Due to Proposition \ref{prop-RR0} the integrand on the left is rapidly decreasing as $|\lambda| \to \infty$, $\lambda \in \Gamma$, so the integral is an entire function of $s \in \C$. Moreover, since the proof of Theorem \ref{thm-trace} was explicitly given in terms of the Schwartz kernels of the operators, it is not difficult to see that the function $K(\lambda)$ is rapidly decreasing as $|\lambda| \to \infty$, $\lambda \in \Gamma$, within the topology of continuous functions on $M \times M$. Hence by Fubini's theorem,
\[
\frac{i}{2\pi} \int_\Gamma \int_M \lambda^{-s} K(\lambda)|_{\mathrm{Diag}} dg\, d\lambda = \frac{i}{2\pi}  \int_M \int_\Gamma \lambda^{-s} K(\lambda)|_{\mathrm{Diag}}   d\lambda\, dg .
\]
Now
\begin{align*}
\frac{i}{2\pi} \int_\Gamma \lambda^{-s} K(\lambda)  d\lambda  & =  \frac{i}{2\pi} \int_\Gamma \lambda^{-s} \Big( \Delta_D(\lambda)^{-1} - \Delta(\lambda)^{-1} - \widehat{\Delta}_D(\lambda)^{-1} + \widehat{\Delta}(\lambda)^{-1}\Big) d\lambda\\
& =  \Delta_D^{-s} - \Delta^{-s} - \widehat{\Delta}_D^{-s} + \widehat{\Delta}^{-s} ,
\end{align*}
where the last equality holds by definition of the complex powers (here we use that $c > 0$ is to the left of all the positive eigenvalues of $\Delta_D,\Delta,\widehat{\Delta}_D,  \widehat{\Delta}$). Thus,
\begin{align*}
\frac{i}{2\pi}  \int_M \int_\Gamma \lambda^{-s} K(\lambda)|_{\mathrm{Diag}}   d\lambda\, dg & =  \int_M \Big( \Delta_D^{-s} - \Delta^{-s} - \widehat{\Delta}_D^{-s} + \widehat{\Delta}^{-s} \Big) \Big|_{\mathrm{Diag}}\, dg\\
&  = \zeta(\Delta_D,s) - \zeta(\Delta,s) - \zeta(\widehat{\Delta}_D,s) + \zeta(\widehat{\Delta},s)
\end{align*}
by definition of the zeta functions.
We conclude that
\begin{multline*}
\frac{i}{2\pi} \int_\Gamma \lambda^{-s} \Tr_Y \big[ R(\lambda)^{-1}R'(\lambda) - \widehat{R}(\lambda)^{-1} \widehat{R}'(\lambda) \big] d\lambda \\ = \zeta(\Delta_D,s) - \zeta(\Delta,s) - \zeta(\widehat{\Delta}_D,s) + \zeta(\widehat{\Delta},s)
\end{multline*}
as required.
\end{proof}

\subsection{A zeta function gluing formula for the $b$-case}

Theorem \ref{thm-Cas} extends in a direct manner to manifolds with cylindrical ends ---  with no changes \emph{except} replacing zeta functions with $b$-zeta functions! Let $M$ be a manifold with cylindrical end decomposed as a union $M = M_- \cup M_+$ of smooth manifolds with boundary and possible cylindrical ends that intersect in a compact codimension one manifold $Y \subseteq M$ such that $Y = \partial M_-  = \partial M_+$; see Figure \ref{fig-mwcyl}. Although both manifolds $M_\pm$ in Figure \ref{fig-mwcyl} have cylindrical ends, one of them could in fact be compact (later in our proofs of the analytic surgery formulas, we shall consider exactly this situation).

\begin{figure} \centering
\includegraphics[width=120mm]{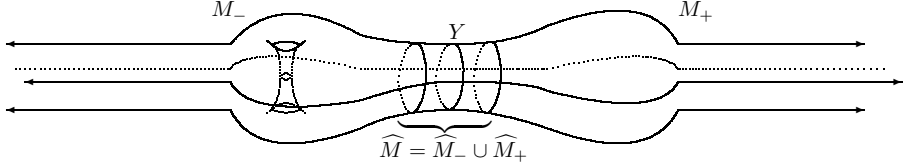} \caption{The partition of the manifold with cylindrical end $M$.} \label{fig-mwcyl}
\end{figure}

Let $\widehat{M} \subseteq M$ be a smooth manifold with boundary, and possibly containing a cylindrical end, containing $Y$ in its interior and suppose that $\widehat{M} = \widehat{M}_- \cup_Y \widehat{M}_+$ where $\widehat{M}_-$ and $\widehat{M}_+$ are manifolds with compact boundaries and possible cylindrical ends such that $Y \subseteq \partial \widehat{M}_\pm$ and $\widehat{M}_- \cap \widehat{M}_+ = Y$. Although both manifolds $\widehat{M}_\pm$ in Figure \ref{fig-mwcyl} are compact, they could in fact have cylindrical ends; however, later in our proof of the analytic surgery formulas $\widehat{M}$ will just be a collar $[-1,1] \times Y$ of $Y$ with $\widehat{M}_- = [-1,0]\times Y$ and $\widehat{M}_+ = [0,1] \times Y$.

Let $\Delta$ be a Laplace-type operator over $M$ that is of product-type over the cylinder, where product-type means that over the cylinder $[0,\infty)_x \times X$ where $X$ is the cross-section (possibly disconnected --- in Figure \ref{fig-mwcyl}, $X$ has two components) of the cylinder, we have
\[
\Delta =  - \partial_x^2 + \Delta_X
\]
where $\Delta_X$ is a Laplace-type operator over $X$.
We assume that the induced cross-sectional Laplace-type operator $\Delta_X$ is invertible. We use $\Delta_\pm$, $\widehat{\Delta}$, $\widehat{\Delta}_\pm$, $R(\lambda)$, $\widehat{R}(\lambda)$, to denote the analogous operators as studied in Section \ref{subsec-DN} but now in the cylindrical end case. Under the invertibility assumption on $\Delta_X$, it is well-known (e.g.\ using the `large' $b$-pseudodifferential calculus of Melrose, also called the `calculus with bounds' \cite[Sec.\ 5.16]{BMeR93}) that the operators $\Delta_\pm$, $\widehat{\Delta}$, $\widehat{\Delta}_\pm$, $R(\lambda)$, $\widehat{R}(\lambda)$ have much of the same properties as in the closed case. For example, each operator $\Delta$, $\Delta_\pm$, $\widehat{\Delta}$, $\widehat{\Delta}_\pm$ is Fredholm and has spectrum consisting only of a set of nonnegative real numbers that is discrete near $0$ and continuous outside of some neighborhood of $0$. (See \cite[Prop.\ 6.27]{BMeR93} --- the bottom of the continuous spectrum of each operator begins at the smallest positive eigenvalue of the cross-sectional Laplace-type operator on the cylindrical end of the manifold over which the operator is defined.) The Dirichlet-to-Neumann map $R(\lambda)$ is defined and analytic for $\lambda \in \C \setminus \big( \mathrm{spec}(\Delta_-) \cup \mathrm{spec}(\Delta_+)\big)$ and is invertible if and only if $\lambda \notin \mathrm{spec}(\Delta)$, in which case
\begin{equation} \label{R-1}
R(\lambda)^{-1} = \gamma \Delta(\lambda)^{-1} \gamma^*.
\end{equation}
In particular, $R(\lambda)^{-1}$ exists for all $\lambda \in \C \setminus [a,\infty)$ for some $a > 0$ except for some discrete subset of $[0,a)$.  A similar statement holds for $\widehat{R}(\lambda)$, and Proposition \ref{prop-RR0} holds. We can now follow the proof of Theorem \ref{thm-Cas}  \emph{word-for-word} in this cylindrical end setting.  We remark that in Section \ref{sec-trace} we emphasized `rapidly decreasing Schwartz kernels' when defining traces and pseudo continuity; the reason for doing so is that $b$-pseudodifferential operators (in the calculus with bounds) are pseudo continuous  (in fact, they are `exponentially pseudo smooth' in the sense that the Schwartz kernel $A(z,z')$ of a $b$-pseudodifferential operator $A$ is smooth and is exponentially decreasing, with all derivatives, off the diagonal in $M \times M$ --- this can be seen by translating the language of the $b$-stretched product in \cite[Sec.\ 5.16]{BMeR93} into variables on infinite cylinders.) Because the trace theorems in Section \ref{sec-trace} were stated for pseudo continuous operators, Step 2 of Theorem \ref{thm-Cas} goes through in the cylindrical end setting without change. The only change in the proof is that wherever there is a zeta function we have to replace it by a $b$-zeta function if the operator is on a manifold with cylindrical end. We summarize our discussion in the following theorem.

\begin{theorem} \label{thm-Cas2}
As meromorphic functions on $\mathbb{C}$ we have
\begin{multline*}
\bz (\Delta_-, s ) + \bz(\Delta_+,s) - \bz ( \Delta,s )     \\ = f(s) + \frac{i}{2\pi} \int_\Gamma \lambda^{-s} \Tr_Y \big[ R(\lambda)^{-1} R'(\lambda) - \widehat{R}(\lambda)^{-1} \widehat{R}'(\lambda) \big] d\lambda ,
\end{multline*}
where
\[
f(s) = \bz(\widehat{\Delta}_-,s) + \bz(\widehat{\Delta}_+,s) - \bz(\widehat{\Delta},s)
\]
and $\Gamma = \{\lambda \in \C\, ;\, \Re \lambda = c\}$ with $c > 0$ sufficiently small.
If any of the manifolds $M_\pm, \widehat{M}, \widehat{M}_\pm$ is compact, we replace $\bz$ with $\zeta$.
\end{theorem}

Theorem \ref{thm-Cas2} fails if we drop the invertibility assumptions on the cross-sectional operators; for example, the operators on the manifolds with cylindrical ends where the cross-sectional Laplacian is not invertible would have continuous spectrum down to the origin so the integral in the gluing formula would not make sense.

\section{Analytic surgery} \label{sec-anasurg}

Using Theorems \ref{thm-Cas} and \ref{thm-Cas2} we prove the analytic surgery theorems in the introduction, modulo some details on Dirichlet-to-Neumann maps which we will present in the Appendix.

\subsection{Stretching a manifold with boundary} \label{sec-stretch1}

We begin with Theorem \ref{thm-main0}. Let $\Delta$ be a Laplace-type operator on $M_0$, a compact Riemannian manifold with boundary $Y : = \partial M_0$. We assume that $M_0$ has a collar neighborhood
\[
M_0 \cong [-1,0]_x \times Y
\]
over which $\Delta = -\D_x^2 + \Delta_Y$ where $\Delta_Y$  \emph{is an invertible  Laplace-type operator on $Y$}. Let $N_r = [0,r] \times Y$ with $r > 1$ and let $M_r$ be the manifold obtained from $M_0$ by attaching the cylinder $N_r$ to $\partial M_0$,
\[
M_r = M_0 \cup_Y N_r ;
\]
see Figure \ref{fig-gluemwb}. $\Delta$ extends in a natural way to $M_r$ and we denote the corresponding Dirichlet Laplacian by $\Delta_{M_r}$.

\begin{figure}[h!] \centering
\includegraphics{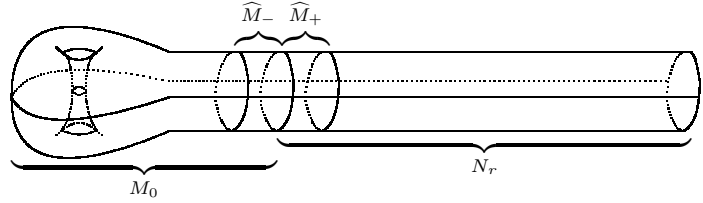} \caption{$M_r = M_- \cup_Y M_+$ where $M_- = M_0$ and $M_+ = N_r$. The submanifold $\widehat{M}_-$ is the original collar $[-1,0] \times Y$ of $M_0$ and $\widehat{M}_+ = [0,1] \times Y \subseteq N_r$.} \label{fig-gluemwb}
\end{figure}

We now use Theorem \ref{thm-Cas} with the partitions shown in Figure \ref{fig-gluemwb}
and get
\begin{multline} \label{lem-01}
\zeta (\Delta_{M_0}, s )  + \zeta (\Delta_{N_r},s )
- \zeta (\Delta_{M_r}, s )  \\ = f(s) + \frac{i}{2\pi} \int_\Gamma \lambda^{-s} \Tr_Y \big[ R_r(\lambda)^{-1}R_r'(\lambda) - \widehat{R}(\lambda)^{-1} \widehat{R}'(\lambda) \big] d\lambda,
\end{multline}
where
\[
f(s) = \zeta(\widehat{\Delta}_{-},s) + \zeta(\widehat{\Delta}_{+},s) - \zeta(\widehat{\Delta},s),
\]
$\Gamma = \{\lambda \in \C\, ;\, \Re \lambda = c\}$ with $c > 0$ small, and where $R_r(\lambda)$ and $\widehat{R}(\lambda)$ denote the Dirichlet-to-Neumann maps for the partitioned manifolds $M = M_- \cup_Y M_+$ and $\widehat{M} = \widehat{M}_- \cup_Y \widehat{M}_+$, respectively,
and, finally, where the subscript $r$ in $R_r(\lambda)$ is to emphasize that the manifold $M_r$ depends on $r$.

\begin{figure}[h!] \centering
\includegraphics{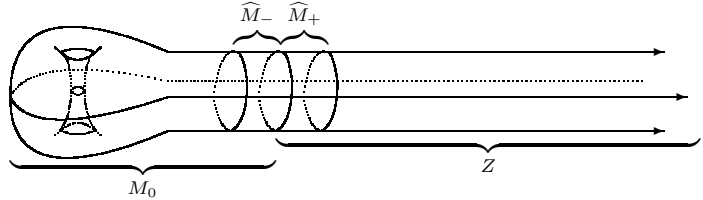} \caption{$M_{\infty} = M_- \cup M_+$ where $M_- = M_0$ and $M_+ = Z$ is a manifold with cylindrical end, and the partitions $\widehat{M}_\pm$ are the same as in Figure \ref{fig-gluemwb}.} \label{fig-gluemwc}
\end{figure}

Let $M_\infty$ denote the manifold $M_0$ with the infinite cylinder $Z := [0,\infty) \times Y$ attached and let $\Delta_{\infty}$ be the canonical extension of $\Delta$ to $M_{\infty}$. We apply Theorem \ref{thm-Cas2} with the decompositions in Figure \ref{fig-gluemwc} and get
\begin{multline} \label{lem-02}
\zeta (\Delta_{M_0}, s )  + \bz (\Delta_{Z},s )
- \bz (\Delta_{M_\infty}, s )  \\ = f(s) + \frac{i}{2\pi} \int_\Gamma \lambda^{-s} \Tr_Y \big[ R_\infty(\lambda)^{-1}R_\infty'(\lambda) - \widehat{R}(\lambda)^{-1} \widehat{R}'(\lambda) \big] d\lambda,
\end{multline}
where $R_\infty(\lambda)$ is the Dirichlet-to-Neumann map for $M_\infty = M_0 \cup_Y Z$ and where the other notations are the same as above.

We now take the combination
$-$\eqref{lem-01} $+$ \eqref{lem-02} and get

\begin{equation} \label{lem-03}
\zeta (\Delta_{M_r}, s ) - \zeta (\Delta_{N_r},s ) + \bz (\Delta_{Z},s )
- \bz (\Delta_{M_\infty}, s ) = \rho(r,s),
\end{equation}
where
\begin{equation} \label{rho}
\rho(r,s) =  \frac{i}{2\pi} \int_\Gamma \lambda^{-s} \Tr_Y \Big(  R_{\infty}(\lambda)^{-1} R_{\infty}'(\lambda) - R_r(\lambda)^{-1} R_r'(\lambda) \Big) \, d\lambda .
\end{equation}

As mentioned earlier, in \cite[Sec.\ 2]{loya04-48-1279} it was proved that $ \bz( \Delta_{Z} ,s) = - \frac14 \zeta(
\Delta_Y ,s)$ and from Proposition \ref{prop-cylinder} we know that
\[
\zeta(\Delta_{N_r},s) = \frac{r}{2} \xi_Y(s) - \frac12 \zeta(\Delta_Y,s) + \kappa (r,s) ,
\]
where $\kappa (r,s)$ is an entire function of $s$ that vanishes exponentially fast as $r \to \infty$ uniformly on compact subsets of $\C$. Thus, \eqref{lem-03} can be written as
\[
\zeta(\Delta_{M_r}, s) - \frac{r}{2} \xi_Y(s) = \bz(\Delta_{M_{\infty}},s) - \frac14 \zeta(\Delta_Y,s) + \kappa (r,s) + \rho(r,s).
\]
So far we have not made the assumption $\ker \Delta_\infty = \{0\}$ that is made in Theorem \ref{thm-main0}; we do so now in order to analyze $\rho(r,s)$.

\begin{proposition} \label{lem-rho}
Assuming $\ker \Delta_\infty = \{0\}$, there is an $r_0 > 0$ such that $\rho(r,s) \in C^\infty((r_0 , \infty) \times \C)$ and is an entire function of $s \in \C$ that vanishes exponentially fast as $r \to \infty$ uniformly on compact subsets of $\C$.
\end{proposition}

This proposition completes the proof of Theorem \ref{thm-main0}. Because the proof of Proposition \ref{lem-rho} is somewhat technical we leave the details to the appendix.

\subsection{Stretching a closed manifold} \label{sec-stretch2}

\begin{figure}[h!] \centering
\includegraphics{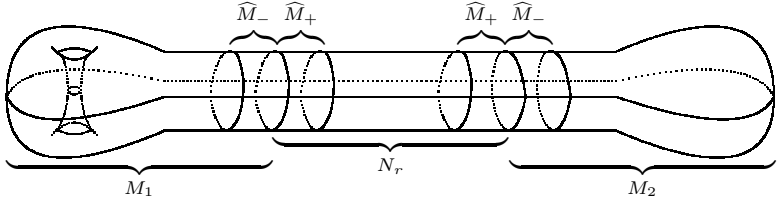} \caption{$M_r= M_- \cup_{\widetilde{Y}} M_+$ where $M_- = M_1 \sqcup M_2$, $M_+ = N_r$ and $\widetilde{Y} = \big(\{-r\} \times Y \big) \sqcup \big( \{r\} \times Y \big) \cong Y \sqcup Y$. The manifolds $\widehat{M}_\pm$ are collar neighborhoods of $\partial M_\pm$.} \label{fig-glueMr}
\end{figure}

We now prove Theorem \ref{thm-main}. Let $\Delta$ be a Laplace-type operator on a closed Riemannian manifold $M$ and let $Y \subseteq M$ be a codimension one submanifold of $M$ decomposing $M$ into two manifolds $M_1$ and $M_2$ with a common boundary. We assume a collar neighborhood
\[
M \cong [-1,1]_x \times Y
\]
over which $\Delta = -\D_x^2 + \Delta_Y$ where $\Delta_Y$ \emph{is an invertible  Laplace-type operator on $Y$}. Let $N_r = [-r,r] \times Y$ where $r > 2$, and split $M$ into the two halves $M_1$ and $M_2$ placing $N_r$ in between:
\[
M_r = M_1 \cup_{\{-r\} \times Y} N_r \cup_{\{r\} \times Y} M_2;
\]
see Figure \ref{fig-glueMr}. Observe that
\[
\widehat{M} \cong [-1,1] \times (Y \sqcup Y) \ , \ \ \widehat{M}_- \cong [-1,0] \times (Y \sqcup Y) \ , \ \ \widehat{M}_+ \cong [0,1] \times (Y \sqcup Y) .
\]

The operator $\Delta$ extends in a natural way to an operator $\Delta_{M_r}$ on $M_r$ and we use Theorem \ref{thm-Cas} with the partitions shown in Figure \ref{fig-glueMr} to obtain
\begin{multline} \label{lem-1}
\zeta (\Delta_{M_1}, s ) +
\zeta (\Delta_{M_2}, s ) - \zeta (\Delta_{M_r},s ) \\ = f(s) + \frac{i}{2\pi} \int_\Gamma \lambda^{-s} \Tr_Y \big[ R_r(\lambda)^{-1}R_r'(\lambda) - \widetilde{R}(\lambda)^{-1} \widetilde{R}'(\lambda) \big] d\lambda,
\end{multline}
where
$f(s) = \zeta(\widehat{\Delta}_{-},s) + \zeta(\widehat{\Delta}_{+},s) - \zeta(\widehat{\Delta},s)$
and
$R_r(\lambda), \widetilde{R}(\lambda)$ are the Dirichlet-to-Neumann maps defined on the dividing hypersurface $Y \sqcup Y$ for $M_r$ and $\widehat{M}$, respectively. Let $N = [-1,1] \times Y$, and split $N$ into two halves $N_- = [-1,0] \times Y$ and $N_+ = [0,1] \times Y$ so that $N = N_- \cup_Y N_+$. Let $\Delta_N$, $\Delta_{N_-}$, $\Delta_{N_+}$ denote the corresponding Dirichlet Laplacians and let $\widehat{R}(\lambda)$ be the Dirichlet-to-Neumann map. Then it follows that
\[
f(s) = 2 g(s)\ ,
\]
where
\[
g(s) = \zeta(\Delta_{N_-},s) + \zeta(\Delta_{N_+},s) - \zeta(\Delta_N,s),
\]
and acting on $C^\infty(Y \sqcup Y) = C^\infty(Y) \oplus C^\infty(Y)$, we have
\[
\widetilde{R}(\lambda) =  \begin{pmatrix}
\widehat{R} (\lambda) & 0 \\
0 & \widehat{R} (\lambda)
\end{pmatrix} .
\]

Now for $i = 1,2$, we put
\[
M_{i,\infty} = M_i \cup_Y Z_i \ , \ \ Z_1 = [0,\infty) \times Y \ , \ \ Z_2 = (-\infty,0] \times Y,
\]
which is a manifold with cylindrical end, and we let $\Delta_{i,\infty}$ denote the canonical extension of $\Delta|_{M_i}$ to $M_{i,\infty}$. Then according to \eqref{lem-02} we have
\begin{multline} \label{lem-2}
\zeta (\Delta_{M_1}, s )  + \bz (\Delta_{Z_1},s )
- \bz (\Delta_{M_{1,\infty}}, s )  \\ = g(s) + \frac{i}{2\pi} \int_\Gamma \lambda^{-s} \Tr_Y \big[ R_{1,\infty}(\lambda)^{-1}R_{1,\infty}'(\lambda) - \widehat{R}(\lambda)^{-1} \widehat{R}'(\lambda) \big] d\lambda
\end{multline}
and
\begin{multline} \label{lem-3}
\zeta (\Delta_{M_2}, s )  + \bz (\Delta_{Z_2},s )
- \bz (\Delta_{M_{2,\infty}}, s )  \\ = g(s) + \frac{i}{2\pi} \int_\Gamma \lambda^{-s} \Tr_Y \big[ R_{2,\infty}(\lambda)^{-1}R_{2,\infty}'(\lambda) - \widehat{R}(\lambda)^{-1} \widehat{R}'(\lambda) \big] d\lambda ,
\end{multline}
where $R_{i,\infty}(\lambda)$ is the Dirichlet-to-Neumann map for $M_{i,\infty} = M_i \cup_Y Z_i$. Recalling that $f(s) = 2 g(s)$,
when we take the combination $-$ \eqref{lem-1} $+$ \eqref{lem-2} $+$ \eqref{lem-3} we obtain
\[
\zeta (\Delta_{M_r}, s ) - \bz (\Delta_{M_{1,\infty}},s ) - \bz (\Delta_{M_{2,\infty}},s )
+ \bz (\Delta_{Z_1}, s ) + \bz(\Delta_{Z_2},s) \\ = \varrho(r,s),
\]
where
\begin{equation} \label{varrho}
\varrho(r,s) = \frac{i}{2\pi} \int_\Gamma \lambda^{-s} \Big(- R_r(\lambda)^{-1} R_r'(\lambda) +R_\infty(\lambda)^{-1} R_\infty ' (\lambda) \Big) d\lambda
\end{equation}
with
\[
R_\infty(\lambda) = \begin{pmatrix}
R_{1,\infty}(\lambda) & 0 \\
0 & R_{2,\infty}(\lambda) \end{pmatrix}.
\]
By \cite[Sec.\ 2]{loya04-48-1279},
\[
\bz (\Delta_{Z_1},s ) = \bz (\Delta_{Z_2},s ) = - \frac14 \zeta(\Delta_Y,s),
\]
so
\[
\zeta (\Delta_{M_r}, s ) - \bz (\Delta_{M_{1,\infty}},s ) - \bz (\Delta_{M_{2,\infty}},s )
-\frac12 \zeta(\Delta_Y,s) \\ = \varrho(r,s).
\]
This formula plus the following theorem, where we now impose the condition $\ker \Delta_{i,\infty} = \{0\}$ for $i = 1,2$, imply Theorem \ref{thm-main}.

\begin{proposition} \label{lem-rho2}
Assuming $\ker \Delta_{i,\infty} = \{0\}$ for $i = 1,2$, there is an $r_0 > 0$ such that $\varrho(r,s)$ belongs to $C^\infty((r_0 , \infty) \times \C)$ and is an entire function of $s \in \C$ such that given any compact subset $K \subseteq \C$ there are constants $c,C > 0$ such that for all $r > r_0$ and $s \in K$,
\[
|\varrho(r,s)| \leq C e^{-cr}.
\]
\end{proposition}

This result is proved in the appendix.


\section*{Acknowledgements}
Our work is supported by National Science Foundation Grants
PHY-0757791 (Baylor) and PHY-0757795 (Binghamton). PL wholeheartedly thanks Jinsung Park for his friendship and teaching him the beauty of gluing formulas.
Part of the work was done while
KK enjoyed the hospitality and partial support of the
Department of Physics and Astronomy of
the University of Oklahoma. Thanks go in particular to Kimball Milton who
made this very pleasant and exciting visit possible.

\appendix
\section{Analysis of Dirichlet-to-Neumann maps} \label{sec-zetacyl}

In this appendix we prove Propositions \ref{lem-rho} and \ref{lem-rho2}.

\subsection{Dirichlet-to-Neumann maps}

We begin by computing the Dirichlet-to-Neumann maps appearing in \eqref{rho}; for the notation in the following proposition see Section \ref{sec-stretch1}.

\begin{proposition} \label{prop-RRinfty}
We have
\begin{enumerate}
\item $R_{\infty}(\lambda) = \mathcal{N}_{M_0}(\lambda) + \sqrt{\Delta_Y(\lambda)}$ where $\mathcal{N}_{M_0}(\lambda)$ is the Dirichlet-to-Neumann map for $M_0$.
\item $R_r(\lambda) =
R_{\infty}(\lambda) + C_r(\lambda)$, where
\[
C_r(\lambda) =  \frac{2 \sqrt{\Delta_Y(\lambda)}}{\Id - e^{-2 r \sqrt{\Delta_Y(\lambda)}}} \, e^{-2 r \sqrt{\Delta_Y(\lambda)}}.
\]
\end{enumerate}
\end{proposition}
\begin{proof}
To prove \textit{(1)} note that by definition \eqref{defR} of the Dirichlet-to-Neumann map,  we have
\[
R_\infty(\lambda) = \mathcal{N}_{M_0}(\lambda) + N_{Z}(\lambda),
\]
so we just have to verify that $N_Z(\lambda) = \sqrt{\Delta_Y(\lambda)}$.
To this end, observe that if $\psi \in C^\infty(Y)$, then
\[
\varphi =  e^{-\sqrt{\Delta_Y(\lambda)}\, x}\, \psi \in C^\infty(Z)
\]
solves $(-\D_x^2 + \Delta_Y(\lambda)) \varphi = 0$ and $\varphi|_{x = 0} = \psi$. The outer unit normal is $-\partial_x$, so
\[
N_Z(\lambda) \psi := - \D_x \varphi \Big|_{x = 0}  = \sqrt{\Delta_Y(\lambda)} \psi,
\]
which completes the proof of \textit{(1)}.

To prove \textit{(2)}, we note that
\begin{align*}
R_r(\lambda) & : = \mathcal{N}_{M_0}(\lambda) + \mathcal{N}_{N_r}(\lambda) \\
& = R_\infty(\lambda) - \sqrt{\Delta_Y(\lambda)} + \mathcal{N}_{N_r}(\lambda)  ,
\end{align*}
so we just have to prove that $\mathcal{N}_{N_r}(\lambda) = \sqrt{\Delta_Y(\lambda)} +  C_r(\lambda)$.

Thus, focusing on $N_r = [0,r] \times Y$ we observe that given $\psi \in C^\infty(Y)$,
\[
\varphi :=  \cosh\big(x \sqrt{\Delta_Y(\lambda)}\big)\, \psi - \frac{\sinh\big( x \sqrt{\Delta_Y(\lambda)} \big)}{\sinh \big(r \sqrt{\Delta_Y(\lambda)}\big)} \cosh\big(r \sqrt{\Delta_Y(\lambda)} \big) \psi
\]
solves $(-\D_x^2 + \Delta_Y(\lambda)) \varphi = 0$ and $\varphi|_{x = 0} = \psi$ and $\varphi|_{x = r} = 0$.
Since
\[
-\D_x \varphi|_{x = 0} = \sqrt{\Delta_Y(\lambda)}  \frac{\cosh\big( r \sqrt{\Delta_Y(\lambda)} \big)}{\sinh \big(r \sqrt{\Delta_Y(\lambda)}\big)} \psi
\]
and $\cosh z = \sinh z + (\cosh z - \sinh z)= \sinh z + e^{-z}$ it follows that
\begin{align*}
\mathcal{N}_{N_r}(\lambda) & = \sqrt{\Delta_Y(\lambda)} +  \sqrt{\Delta_Y(\lambda)}  \frac{e^{- r \sqrt{\Delta_Y(\lambda)}}}{\sinh \big(r \sqrt{\Delta_Y(\lambda)}\big)} \\
& = \sqrt{\Delta_Y(\lambda)} +  C_r  (\lambda),
\end{align*}
exactly what we wanted.
\end{proof}

Using \emph{(2)} we can prove Proposition \ref{lem-rho} but in order to do so we need to understand the $r,|\lambda| \to \infty$ behavior of $C_r(\lambda)$, which we consider next.

\subsection{Rapidly decreasing parameter-dependent operators}

Let $\Lambda \subseteq \C$ be a sectorial region, which recall means that  outside some neighborhood of the origin, $\Lambda$ equals a sector (solid angle). We define the space $\widetilde{\Psi}^{-\infty}_{\Lambda}(Y)$ as the space of $\Psi^{-\infty}_\Lambda(Y)$-valued Schwartz functions on $[1,\infty)_r$ that are exponentially decreasing in $r$; more precisely, $K_r(\lambda) \in \widetilde{\Psi}^{-\infty}_{\Lambda}(Y)$ means that for some $\varepsilon > 0$, we have
\[
e^{\varepsilon r} K_r(\lambda) \in \mathcal{S}([1,\infty)_r  ; \Psi^{-\infty}_\Lambda(Y)),
\]
where the space on the right consists of $\Psi^{-\infty}_\Lambda(Y)$-valued Schwartz functions on $[1,\infty)_r$. Equivalently, we can put $\Lambda$ as another parameter space and write
\[
e^{\varepsilon r} K_r(\lambda) \in \mathcal{S}([1,\infty)_r \times \Lambda ; \Psi^{-\infty} (Y)).
\]
Here is a useful lemma concerning this space of operators.

\begin{lemma} \label{lem-ideal}
The space $\widetilde{\Psi}^{-\infty}_{\Lambda}(Y)$ is an ideal in the space of parameter dependent operators in the sense that if $K_r(\lambda) \in \widetilde{\Psi}^{-\infty}_{\Lambda}(Y)$ and $A(\lambda)$ is parameter dependent in $\lambda \in \Lambda$ of any given weight, then $A(\lambda) K_r(\lambda), K_r(\lambda) A(\lambda) \in \widetilde{\Psi}^{-\infty}_{\Lambda}(Y)$.
\end{lemma}
\begin{proof} Since the spaces of parameter dependent operators are closed under taking adjoints, we just have to check the result for $A(\lambda) K_r(\lambda)$ (since the result for $K_r(\lambda) A(\lambda)$ would follow from the result for $A(\lambda)^* K_r(\lambda)^*$ by taking adjoints).  By definition of parameter dependent operators (see the discussion around \eqref{asymbol}), if the Schwartz kernel of $A(\lambda)$ vanishes near the diagonal in $Y \times Y$, then $A(\lambda)$ is an element of $\Psi^{-\infty}_\Lambda(Y)$ and in this case it is easy to prove that $A(\lambda) K_r(\lambda) \in \widetilde{\Psi}^{-\infty}_{\Lambda}(Y)$. Thus, we may assume that the Schwartz kernel of $A(\lambda)$ (and $K_r(\lambda)$) is supported on some coordinate patch $\U \times \U$ where $\U$ is a coordinate patch on $Y$, in which case we can write
\[
A(\lambda,y,y') = \int e^{i(y - y') \cdot \eta} \, a(\lambda,y,\eta) \dbar \eta ,
\]
where $a(\lambda, y,\eta)$ satisfies the symbol estimates \eqref{asymbol}.
Therefore with $B_r(\lambda) : = A(\lambda)K_r(\lambda)$, we have
\[
B_r(\lambda,y,z) = \int e^{i y \cdot \eta} \, b(r,\lambda,y,z,\eta) \dbar \eta
\]
where
\[
b(r,\lambda,y,z,\eta) = a(\lambda,y,\eta) \int e^{- i y' \cdot \eta} K_r(\lambda,y',z)\, dy'.
\]
Since $K_r(\lambda,y',z)$ is smooth in all variables and rapidly decreasing as $r,|\lambda| \to \infty$, $b(r,\lambda,y,z,\eta)$ has the same properties and, by well-known results on the Fourier transform, is rapidly decreasing as $|\eta| \to \infty$. It follows that $B_r(\lambda,y,z)$ has the same properties as $K_r(\lambda,y',z)$. This completes our proof.
\end{proof}

\begin{figure} \centering
\includegraphics{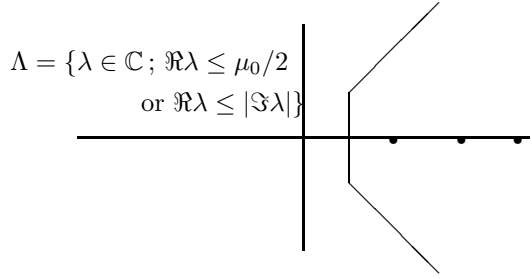} \caption{The sectorial region $\Lambda$. The dots are the eigenvalues of $\Delta_Y$ and $\mu_0 > 0$ is the smallest one.} \label{fig-Gamma2}
\end{figure}

\begin{lemma} \label{lem-estimate}
If $\Lambda$ is the region in Figure \ref{fig-Gamma2} and $a_0 > 0$, then there is a constant $c > 0$ such that for any $a \in [a_0,\infty)$, $\lambda  \in \Lambda$ and $r \geq 1$ we have
\[
\left| \frac{1}{1 - e^{-2 r \sqrt{a - \lambda}}} \right| \leq \frac{1}{1 - e^{- c \sqrt{a_0}}}\ \ , \quad \left| e^{-2 r \sqrt{a - \lambda}}  \right| \leq e^{- cr} \cdot e^{-c \sqrt{|\lambda|}} \cdot  e^{- c \sqrt{a}}.
\]
\end{lemma}
\begin{proof}  We first claim there is a $b > 0$ such that for any $a \in (0,\infty)$ and $\lambda  \in \Lambda$, we have
\begin{equation}
\Re \sqrt{a - \lambda} \geq b |\lambda| + b \sqrt{a}, \label{starstar}
\end{equation}
where we will see that $b = \frac{1}{4} \sqrt{1 - \frac{1}{\sqrt{2}}}$ works.
Let $\lambda = x + iy \in \Lambda$ and write $a - \lambda =  |a - \lambda| \, e^{\pm i\theta}$ where $\cos \theta = (a - x)/\sqrt{(a - x)^2 + y^2}$ and where the $\pm$ depends on the sign of $y$. Thus, $\sqrt{a - \lambda} = \sqrt{|a - \lambda|}\, e^{\pm i \theta/2}$, so
\[
\Re \sqrt{a - \lambda} = \sqrt{|a - \lambda|}\, \cos (\theta/2).
\]
By the half-angle formula,
\[
\cos (\theta/2) = \frac{1}{\sqrt{2}}\sqrt{1 + \cos \theta} = \frac{1}{\sqrt{2}} \sqrt{1 + \frac{a - x}{\sqrt{(a - x)^2 + y^2}}}.
\]
An elementary exercise shows that for $\lambda = x + iy \in \Lambda$,
\[
\frac{a - x}{\sqrt{(a - x)^2 + y^2}} \geq - \frac{1}{\sqrt{2}} ,
\]
and $\sqrt{2}\sqrt{|a - \lambda|} \geq \sqrt{|\lambda|}$ and $\sqrt{2} \sqrt{|a - \lambda|} \geq \sqrt{a}$.
Thus,
\begin{align*}
\Re \sqrt{a - \lambda} \geq  \sqrt{|a - \lambda|}\, \frac{1}{\sqrt{2}} \sqrt{1 - \frac{1}{\sqrt{2}}} & = 2^{3/2} b \sqrt{|a - \lambda|}\\
& = 2^{1/2} b \sqrt{|a - \lambda|} + 2^{1/2} b \sqrt{|a - \lambda|}\\ & \geq b \sqrt{|\lambda|} + b \sqrt{a}.
\end{align*}
This proves our claim.

Using (\ref{starstar}), one can verify that for all $a \in [a_0,\infty)$, $\lambda  \in \Lambda$ and $r \geq 1$ we have
\[
\left| \frac{1}{1 - e^{-2 r \sqrt{a - \lambda}}} \right| \leq \frac{1}{1 - e^{-2 b \sqrt{a_0}}}\ \ , \quad \left| e^{-2 r \sqrt{a - \lambda}}  \right| \leq e^{-2 r b (\sqrt{|\lambda|} + \sqrt{a})} .
\]
Note that there is a constant $c \in \R$ with $0 < c \leq 2 b$ such that for all $u \geq 1$ and $v \geq \sqrt{a_0}$ we have $c(u + v) \leq 2 b u v$ (just take $c = 2b/(1 + 1/\sqrt{a_0})$). Putting $u = r \geq 1$ and $v = \sqrt{|\lambda|} + \sqrt{a} \geq \sqrt{a_0}$ shows that
\[
e^{-2 r b (\sqrt{|\lambda|} + \sqrt{a})}  \leq e^{-c r} \cdot e^{-c \sqrt{|\lambda|}} \cdot e^{-c \sqrt{a}}.
\]
This completes our proof.
\end{proof}

\begin{proposition} \label{propCr} If $\Lambda$ is the region in Figure \ref{fig-Gamma2}, then
$C_r(\lambda) \in \widetilde{\Psi}^{-\infty}_{\Lambda}(Y)$.
\end{proposition}
\begin{proof}
Let $\{\mu_k \}$ be the eigenvalues of $\Delta_Y$, with $\mu_0$ the smallest one, and let $\{\varphi_k\}$ be the corresponding orthonormal eigenvectors. Then it follows directly from the properties of the topology on $\Psi^{-\infty}(Y)$ that a parameter dependent operator $K_r(\lambda)$ on $L^2(M)$ defines an element $K_r(\lambda) \in \mathcal{S}([1,\infty)_r \times \Lambda ; \Psi^{-\infty} (Y))$ if and only if for each $k,\ell$, the function
\[
\langle K_r(\lambda) \varphi_k, \varphi_\ell \rangle
\]
is smooth in $(r,\lambda) \in [1,\infty)_r \times \Lambda$ and rapidly decreasing, with all derivatives in $(r,\lambda)$, as $r,|\lambda|,k,\ell \to \infty$. Indeed, one can prove this from the Fourier series representation of the Schwartz kernel of $K_r(\lambda)$:
\[
K_r(\lambda) = \sum_{k,\ell} \langle K_r(\lambda) \varphi_k, \varphi_\ell \rangle\, \varphi_k \otimes \varphi_\ell.
\]
In our situation the Schwartz kernel $C_r(\lambda)$ is of the form
\[
C_r(\lambda) = \sum_k \frac{2 \sqrt{\mu_k - \lambda}}{1 - e^{-2 r \sqrt{\mu_k - \lambda}}} \, e^{-2 r \sqrt{\mu_k - \lambda}}\, \varphi_k \otimes \varphi_k,
\]
so we just have to prove that
\[
\frac{2 \sqrt{\mu_k - \lambda}}{1 - e^{-2 r \sqrt{\mu_k - \lambda}}} \, e^{-2 r \sqrt{\mu_k - \lambda}}
\]
is rapidly decreasing, with all derivatives in $(r,\lambda) \in [1,\infty) \times \Lambda$, as $r,|\lambda|,k \to \infty$. Observe that any such derivative is a linear combination of terms of the form
\begin{equation} \label{fcoeff}
\frac{ \big(\sqrt{\mu_k - \lambda} \big)^x}{\big(1 - e^{-2 r \sqrt{\mu_k - \lambda}}\big)^y} \, e^{-2 z r \sqrt{\mu_k - \lambda}} ,
\end{equation}
where $x \in \Z$ and $y,z \in \N$.

Now according to Lemma \ref{lem-estimate} there is a constant $c > 0$ such that for all $\lambda \in \Lambda$, $r \geq 1$, and all $k$,
\[
\left| \frac{1}{1 - e^{-2 r \sqrt{\mu_k - \lambda}}}  \right| \leq C\ , \quad \text{where}\ \ C : = \frac{1}{1 - e^{-c \sqrt{\mu_0}}},
\]
and
\[
\left| e^{-2 r \sqrt{\mu_k - \lambda}} \right| \leq e^{- cr} \cdot e^{-c \sqrt{|\lambda|}} \cdot  e^{- c \sqrt{\mu_k}}.
\]
Hence, for all $\lambda \in \Lambda$, $r \geq 1$, and all $k$,
\[
\left| \eqref{fcoeff} \right| \leq C^y \big( \sqrt{|\mu_k - \lambda|} \big)^x e^{- czr} \cdot e^{-c z \sqrt{|\lambda|}} \cdot  e^{- c z \sqrt{\mu_k}} ,
\]
which is rapidly decreasing as $r,|\lambda|,k \to \infty$. Thus, $C_r(\lambda) \in \mathcal{S}([1,\infty)_r \times \Lambda ; \Psi^{-\infty} (Y))$. Moreover, because of the $e^{-czr}$ term in the previous inequality it follows that
\[
e^{\varepsilon r} C_r(\lambda) \in \mathcal{S}([1,\infty)_r \times \Lambda ; \Psi^{-\infty} (Y))
\]
where $\varepsilon = c/2$. This completes our proof.
\end{proof}

\subsection{Proof of Proposition \ref{lem-rho}.}

Assuming that $\ker \Delta_Y = \{0\}$ and $\ker \Delta_\infty = \{0\}$, we need to prove that
there is an $r_0 > 0$ such that
\[
\rho(r,s) =  \frac{i}{2\pi} \int_\Gamma \lambda^{-s} \Tr_Y \Big(  R_{\infty}(\lambda)^{-1} R_{\infty}'(\lambda) - R_r(\lambda)^{-1} R_r'(\lambda) \Big) \, d\lambda
\]
belongs to $C^\infty((r_0 , \infty) \times \C)$ and is an entire function of $s \in \C$ that vanishes exponentially fast as $r \to \infty$ uniformly on compact subsets of $\C$. Recall that $\Gamma = \{\lambda \in \C\, ;\, \Re \lambda = c\}$ where $c > 0$ is chosen such that $R_{\infty}(\lambda)^{-1}$ and $R_r(\lambda)^{-1}$ are defined for $\lambda \in (0,c]$ (see Figure \ref{fig-Gamma}).

\begin{proof}
We know that $R_r(\lambda) =
R_{\infty}(\lambda) + C_r(\lambda)$ where $C_r(\lambda)$ is the operator given in Proposition \ref{prop-RRinfty}. Hence,
\begin{equation} \label{RrRinfty}
R_r(\lambda) = (\Id + G_r(\lambda)) R_\infty(\lambda),
\end{equation}
where $G_r(\lambda) = C_r(\lambda) R_\infty(\lambda)^{-1}$, provided that $R_\infty(\lambda)$ is invertible. Now comes the assumption $\ker \Delta_\infty = \{0\}$. Since $\Delta_\infty$ is invertible it follows that (see the discussion around \eqref{R-1}) $R_\infty(0)^{-1}$ exists and even more,  $R_\infty(\lambda)^{-1}$ exists for all $\lambda \in \C \setminus (a,\infty)$ for some $a > 0$. Let
\[
\Lambda = \{ \lambda \in \mathbb{C}\, ;\, \Re \lambda \leq \varepsilon \ \text{or}\ \Re \lambda \leq |\Im \lambda|\},
\]
where $\varepsilon$ is the minimum of $a$ or $\mu_0/2$ with $\mu_0$ is the smallest positive eigenvalue of $\Delta_Y$. Then $R_\infty(\lambda)^{-1}$ is analytic for all $\lambda \in \Lambda$. Moreover, we already know that $R_\infty(\lambda)^{-1}$ is parameter dependent of weight $-1$ (cf.\ Proposition \ref{prop-R}), so by Proposition \ref{propCr}, $C_r(\lambda) \in \widetilde{\Psi}^{-\infty}_{\Lambda}(Y)$. Hence, by Lemma \ref{lem-ideal}, we have $G_r(\lambda) \in \widetilde{\Psi}^{-\infty}_{\Lambda}(Y)$. Thus there is an $r_0 > 0$ such that for all $r > r_0$, the operator norm of $G_r(\lambda)$ on $L^2(Y)$ is less then $1/2$ for all $r > r_0$ and $\lambda \in \Lambda$. Thus, from \eqref{RrRinfty} it follows that $R_r(\lambda)^{-1}$ exists for all $r > r_0$ and $\lambda \in \Lambda$, and
\[
R_r(\lambda)^{-1} = R_\infty(\lambda)^{-1} (\Id + G_r(\lambda))^{-1}.
\]
Now
\[
R_r'(\lambda) = G_r'(\lambda) R_\infty(\lambda) + (\Id + G_r(\lambda)) R_\infty'(\lambda)
\]
so
\[
R_{\infty}(\lambda)^{-1} R_{\infty}'(\lambda)
-  R_r(\lambda)^{-1} R_r'(\lambda)= - R_\infty(\lambda)^{-1} (\Id + G_r(\lambda))^{-1} G_r'(\lambda) R_\infty(\lambda).
\]
From this we see that for all $r > r_0$,
\begin{equation} \label{rhors}
\rho(r,s) = - \frac{i}{2\pi} \int_\Gamma \lambda^{-s} \Tr_Y \Big((\Id + G_r(\lambda))^{-1} G_r'(\lambda) \Big)   \, d\lambda,
\end{equation}
where we can put $\Gamma = \{ \lambda \in \C\, ;\, \Re \lambda = \varepsilon\}$, a contour which is independent of $r > r_0$. Finally, recalling that $G_r(\lambda) \in \widetilde{\Psi}^{-\infty}_\Lambda(Y)$ and the definition of the space $\widetilde{\Psi}^{-\infty}_\Lambda(Y)$ it follows immediately that $\rho(r,s)$ is, for $r > r_0$, an entire function of $s \in \C$ that vanishes exponentially fast as $r \to \infty$ uniformly on compact subsets of $\C$.
\end{proof}

The assumption $\ker \Delta_\infty = \{0\}$ is important for the following reason. If $\ker \Delta_\infty \ne \{0\}$ then $R_\infty(\lambda)^{-1}$ would have a pole at $\lambda = 0$. Thus, $G_r(\lambda) = C_r(\lambda) R_\infty(\lambda)^{-1}$ may have, for any $r > 0$ no matter how large, an arbitrary large norm for small $\lambda > 0$. Hence for any $r > 0$, $R_r(\lambda)^{-1}$ may fail to exist for some $\lambda > 0$ sufficiently small. Now recall that $\Gamma = \{\lambda \in \C\, ;\, \Re \lambda = c\}$ is such that  $R_r(\lambda)^{-1}$ must be defined for $\lambda \in (0,c]$. Thus, if $\ker \Delta_\infty \ne \{0\}$, then it is possible that $c$ would depend on $r$. This would make the analysis of \eqref{rhors} highly nontrivial.

\subsection{Proof of Proposition \ref{lem-rho2}}

To prove Proposition \ref{lem-rho2} we first compute the Dirichlet-to-Neumann maps appearing in \eqref{varrho}; for the notation in the following proposition see Section \ref{sec-stretch2}.

\begin{proposition} \label{propRrinfty}
We have
$\displaystyle
R_r(\lambda) =
R_\infty(\lambda) + T_r(\lambda)
$, where
\[
R_\infty(\lambda) = \begin{pmatrix}
R_{1,\infty}(\lambda) & 0 \\
0 & R_{2,\infty}(\lambda)
\end{pmatrix}
\]
and
\[
T_r(\lambda) = \frac{\sqrt{\Delta_Y(\lambda)}}{\sinh (2r \sqrt{\Delta_Y(\lambda)})} \begin{pmatrix} e^{-2r \sqrt{\Delta_Y(\lambda)}} & - \Id \\
- \Id & e^{-2r\sqrt{\Delta_Y(\lambda)}}\end{pmatrix} .
\]
\end{proposition}
\begin{proof}
We first claim that for $N_r = [-r,r] \times Y$, we have
\[
\mathcal{N}_{N_r} =
\begin{pmatrix}
\sqrt{\Delta_Y(\lambda)} & 0 \\ 0 & \sqrt{\Delta_Y(\lambda)}
\end{pmatrix}
+
T_r(\lambda).
\]
To prove this we first note that given $(\psi_1 , \psi_2) \in C^\infty(Y \sqcup Y)$,
\begin{multline*}
\varphi = \cosh\big((x - r) \sqrt{\Delta_Y(\lambda)}\big)\, \psi_2\\
- \frac{\sinh\big( (x - r) \sqrt{\Delta_Y(\lambda)} \big)}{\sinh \big(2 r \sqrt{\Delta_Y(\lambda)}\big)} \Big(\psi_1 - \cosh\big(2r \sqrt{\Delta_Y(\lambda)} \psi_2 \big) \Big)
\end{multline*}
solves $(-\D_x^2 + \Delta_Y(\lambda)) \varphi = 0$ and $\varphi|_{x = -r} = \psi_1$ and $\varphi|_{x = r} = \psi_2$.
Algebra shows that
\[
-\D_x \varphi|_{x = -r} = \sqrt{\Delta_Y(\lambda)} \psi_1 + \frac{\sqrt{\Delta_Y(\lambda)}}{\sinh (2r \sqrt{\Delta_Y(\lambda)})} \Big( - \psi_2 + e^{-2r \sqrt{\Delta_Y(\lambda)}} \psi_1 \Big)
\]
and
\[
\D_x \varphi|_{x = r} = \sqrt{\Delta_Y(\lambda)} \psi_1 + \frac{\sqrt{\Delta_Y(\lambda)}}{\sinh (2r \sqrt{\Delta_Y(\lambda)})} \Big( - \psi_1 + e^{-2r \sqrt{\Delta_Y(\lambda)}} \psi_2 \Big).
\]
Hence,
\[
\mathcal{N}_{N_r} =
\begin{pmatrix}
\sqrt{\Delta_Y(\lambda)} & 0 \\ 0 & \sqrt{\Delta_Y(\lambda)}
\end{pmatrix}
+
\frac{\sqrt{\Delta_Y(\lambda)}}{\sinh (2r \sqrt{\Delta_Y(\lambda)})}
\begin{pmatrix} e^{-2r \sqrt{\Delta_Y(\lambda)}} & - \Id \\
-\Id & e^{-2r \sqrt{\Delta_Y(\lambda)}}
\end{pmatrix},
\]
which proves our claim.

We now prove the proposition. Indeed, by definition of $R_r(\lambda)$ we have
\[
R_r(\lambda) =
\begin{pmatrix}
\mathcal{N}_{M_1} & 0 \\ 0 & \mathcal{N}_{M_2}
\end{pmatrix} + \mathcal{N}_{N_r},
\]
and by our formula for $\mathcal{N}_{N_r}$ computed above, we see that
\[
R_r(\lambda) =  \begin{pmatrix}
\mathcal{N}_{M_1} + \sqrt{\Delta_Y(\lambda)} & 0 \\ 0 & \mathcal{N}_{M_2} + \sqrt{\Delta_Y(\lambda)}
\end{pmatrix} +
T_r(\lambda),
\]
which when combined with \textit{(1)} of Proposition \ref{prop-RRinfty} proves the result.
\end{proof}

The proof of Proposition \ref{propCr} can be adapted to prove the following.

\begin{proposition} \label{propKr} If $\Lambda$ is the region in Figure \ref{fig-Gamma2}, then
$T_r(\lambda) \in \widetilde{\Psi}^{-\infty}_{\Lambda}(Y)$.
\end{proposition}

Now that we have Propositions \ref{propRrinfty} and \ref{propKr} we can use them
to prove Proposition \ref{lem-rho2} in an almost identical way as we proved Proposition \ref{lem-rho}. Thus, we omit the similar details.


\begin{thebibliography}{10}

\bibitem{blau88-209-209}
S.K. Blau, M.~Visser, and A.~Wipf, \emph{Determinants of conformal wave
  operators in four dimensions}, Phys. Lett. \textbf{B209} (1988), 209--213.

\bibitem{blau89-4-1467}
\bysame, \emph{Determinants, {D}irac operators, and one loop physics}, Int. J.
  Mod. Phys. \textbf{A4} (1989), 1467--1484.

\bibitem{BoM-GeB-KiK-ElE96}
M.~Bordag, B.~Geyer, K.~Kirsten, and E.~Elizalde, \emph{Zeta function
  determinant of the {L}aplace operator on the {$D$}-dimensional ball}, Comm.
  Math. Phys. \textbf{179} (1996), no.~1, 215--234.

\bibitem{BoM-KiK-ElE96}
M.~Bordag, K.~Kirsten, and E.~Elizalde, \emph{Heat kernel coefficients of the
  {L}aplace operator on the {D}-dimensional ball}, J. Math. Phys. \textbf{37}
  (1996), no.~1, 895--916.

\bibitem{bord01-353-1}
M.~Bordag, U.~Mohideen, and V.M. Mostepanenko, \emph{New developments in the
  {C}asimir effect}, Phys. Rept. \textbf{353} (2001), 1--205.

\bibitem{BoM-KiK-DoS96}
M.~Bordag, K.~Kirsten, and S.~Dowker, \emph{Heat-kernels and
  functional determinants on the generalized cone}, Comm. Math. Phys.
  \textbf{182} (1996), no.~2, 371--393.

\bibitem{bran92-149-241}
T.P. Branson, S.-Y.A. Chang, and P.C. Yang, \emph{Estimates and extremals for
  zeta function determinants on four-manifolds}, Commun. Math. Phys.
  \textbf{149} (1992), 241--262.

\bibitem{bran94-344-479}
T.P. Branson and P.B. Gilkey, \emph{The functional determinant of a
  4-dimensional boundary-value problem}, Trans. Am. Math. Soc. \textbf{344}
  (1994), 479--531.

\bibitem{buch85-162-92}
I.L. Buchbinder, S.D. Odintsov, and I.L. Shapiro, \emph{Nonsingular
  cosmological model with torsion induced by vacuum quantum effects}, Phys.
  Lett. \textbf{B162} (1985), 92--96.

\bibitem{buch89-12-1}
\bysame, \emph{Renormalization group approach to quantum field theory in curved
  space-time}, Riv. Nuovo Cim. \textbf{12} (1989), 1--112.

\bibitem{BuD-FrL-KaT92}
D.~Burghelea, L.~Friedlander, and T.~Kappeler, \emph{Meyer-{V}ietoris type
  formula for determinants of elliptic differential operators}, J. Funct. Anal.
  \textbf{107} (1992), no.~1, 34--65.

\bibitem{CarG02}
G.~Carron, \emph{Determinant relatif et la fonction {X}i}, Am. J. Math.
  \textbf{124} (2002), 307--352.

\bibitem{IDaiX06}
X.~Dai, \emph{Eta invariants for manifold with boundary}, Analysis,
  geometry and topology of elliptic operators, World Sci. Publ., Hackensack,
  NJ, 2006, pp.~141--172.

\bibitem{dett92-377-252}
A.~Dettki and A.~Wipf, \emph{Finite size effects from general covariance and
  {W}eyl anomaly}, Nucl. Phys. \textbf{B377} (1992), 252--280.

\bibitem{doug91-142-139} R.G.~Douglas and
K.~Wojciechowski, \emph{Adiabatic limits of the
$\eta$-invariants. {T}he odd-dimensional
{A}tiyah-{P}atodi-{S}inger problem}, Comm. Math.
Phys. \textbf{142} (1991), no.~1, 139--168.

\bibitem{dowk86-33-3150}
J.S. Dowker, \emph{Conformal transformation of the effective action}, Phys.
  Rev. \textbf{D33} (1986), 3150--3151.

\bibitem{dowk94-162-633}
\bysame, \emph{Effective action in spherical domains}, Commun. Math. Phys.
  \textbf{162} (1994), 633--648.

\bibitem{dowk78-11-895}
J.S. Dowker and G.~Kennedy, \emph{Finite temperature and boundary effects in
  static space-times}, J. Phys. A: Math. Gen. \textbf{11} (1978), 895--920.

\bibitem{dowk88-38-3327}
J.S. Dowker and J.P. Schofield, \emph{High temperature expansion of the free
  energy of a massive scalar field in a curved space}, Phys. Rev. \textbf{D38}
  (1988), 3327--3329.

\bibitem{dowk89-327-267}
\bysame, \emph{Chemical potentials in curved space}, Nucl. Phys. \textbf{B327}
  (1989), 267--284.

\bibitem{dowk90-31-808}
\bysame, \emph{Conformal transformations and the effective action in the
  presence of boundaries}, J. Math. Phys. \textbf{31} (1990), 808--818.

\bibitem{DuiJ81}
J.~J. Duistermaat, \emph{On operators of trace class in {$L\sp{2}(X,\,\mu )$}},
  Proc. Indian Acad. Sci. Math. Sci. \textbf{90} (1981), no.~1, 29--32.

\bibitem{eliz94b}
E.~Elizalde, S.D. Odintsov, A.~Romeo, A.A. Bytsenko, and S.~Zerbini, \emph{Zeta
  regularization techniques with applications}, World Scientific, Singapore,
  1994.

\bibitem{emig07-99-170403}
T.~Emig, N.~Graham, R.L. Jaffe, and M.~Kardar, \emph{{Casimir forces between
  arbitrary compact objects}}, Phys. Rev. Lett. \textbf{99} (2007), 170403.

\bibitem{fran35-173-245}
W.~Franz, \emph{{\"U}ber die {T}orsion einer {{\"U}}berdeckung}, J. Reine
  Angew. Math. \textbf{173} (1935), 245--254.

\bibitem{BGoI-GoS-KrN00}
I.~Gohberg, S.~Goldberg, and N.~Krupnik, \emph{Traces and
  determinants of linear operators}, Operator Theory: Advances and
  Applications, vol. 116, Birkh\"auser Verlag, Basel, 2000.

\bibitem{hass98-6-255}
A.~Hassell, \emph{Analytic surgery and analytic torsion}, Comm. Anal. Geom.
  \textbf{6} (1998), 255--289.

\bibitem{hass95-3-115}
A.~Hassell, R.R. Mazzeo, and R.B. Melrose, \emph{Analytic surgery and the
accumulation of eigenvalues}, Comm. Anal. Geom. \textbf{3} (1995),
  115--222.

\bibitem{hass97-36-1055}
\bysame, \emph{A signature formula for
  manifolds with corners of codimension two}, Topology \textbf{36} (1997),
  1055--1075.

\bibitem{hass99-18-971}
A.~Hassell and S.~Zelditch, \emph{Determinants of {L}aplacians in
  exterior domains}, Internat. Math. Res. Notices (1999), no.~18, 971--1004.

\bibitem{hert05-95-250402}
M.P. Hertzberg, R.L. Jaffe, M.~Kardar, and A.~Scardicchio, \emph{{Attractive
  Casimir Forces in a Closed Geometry}}, Phys. Rev. Lett. \textbf{95} (2005),
  250402.

\bibitem{kirs91-8-2239}
K.~Kirsten, \emph{Grand thermodynamic potential in a static space-time with
  boundary}, Class. Quantum Grav. \textbf{8} (1991), 2239--2255.

\bibitem{BKirK01}
K.~Kirsten, \emph{Spectral functions in mathematics and physics}, Chapman \&
  Hall/CRC Press, Boca Raton, 2001.

\bibitem{KiK-LoP-PaJII06}
K.~Kirsten, P.~Loya, and J.~Park, \emph{Zeta functions of {D}irac and
  {L}aplace-type operators over finite cylinders}, Ann. Physics \textbf{321}
  (2006), no.~8, 1814--1842.

\bibitem{KorT89}
T.~W. K{\"o}rner, \emph{Fourier analysis}, second ed., Cambridge University
  Press, Cambridge, 1989.

\bibitem{lee03-355-4093}
Y.~Lee, \emph{Burghelea-friedlander-kappeler's gluing formula for the
  zeta-determinant and its applications to the adiabatic decompositions of the
  zeta-determinant and the analytic torsion}, Trans. Amer. Math. Soc.
  \textbf{355} (2003), 4093--4110.

\bibitem{ILeeY05}
\bysame, \emph{Asymptotic expansion of the zeta-determinant of an
  invertible {L}aplacian on a stretched manifold}, Spectral geometry of
  manifolds with boundary and decomposition of manifolds, Contemp. Math., vol.
  366, Amer. Math. Soc., Providence, RI, 2005, pp.~95--108.

\bibitem{loya04-48-1279}
P.~Loya and J.~Park, \emph{Decomposition of the {$\zeta$}-determinant
  for the {L}aplacian on manifolds with cylindrical end}, Illinois J. Math.
  \textbf{48} (2004), no.~4, 1279--1303.

\bibitem{loya05-15-285}
\bysame, \emph{On the gluing problem for {D}irac operators on manifolds with
  cylindrical ends}, J. Geom. Analysis \textbf{15} (2005), 285--319.

\bibitem{mazz95-5-14}
R.~Mazzeo and R.B. Melrose, \emph{Analytic surgery and the eta invariant},
  Geom. Funct. Anal. \textbf{5} (1995), 14--75.

\bibitem{BMeR93}
R.B. Melrose, \emph{The {Atiyah}-{Patodi}-{Singer} {Index} {Theorem}}, A.K.
  Peters, Wellesley, 1993.

\bibitem{UMelRglobal}
\bysame, \emph{From microlocal to global analysis}, Available at
  \texttt{http://www-math.mit.edu/$\sim$rbm}.

\bibitem{milt01b}
K.A. Milton, \emph{The {C}asimir effect: Physical manifestations of zero-point
  energy}, River Edge, USA: World Scientific, 2001.

\bibitem{milt04-37-209}
\bysame, \emph{The {C}asimir effect: Recent controversies and progress}, J.
  Phys. \textbf{A37} (2004), R209--R277.

  \bibitem{mull98-192-309}
W.~M{\"u}ller, \emph{Relative zeta functions, relative determinants and
  scattering theory}, Comm. Math. Phys. \textbf{192} (1998), no.~2, 309--347.

\bibitem{MulJ-MulW06}
J.~M{\"u}ller and W.~M{\"u}ller, \emph{Regularized determinants of
  {L}aplace-type operators, analytic surgery, and relative determinants}, Duke
  Math. J. \textbf{133} (2006), no.~2, 259--312.

\bibitem{osgo88-80-148}
B.~Osgood, R.~Phillips, and P.~Sarnak, \emph{Extremals of determinants of
  {L}aplacians}, J. Funct. Anal. \textbf{80} (1988), 148--211.

\bibitem{PaJ-WoK00}
J.~Park and K.~P. Wojciechowski, \emph{Relative
  {$\zeta$}-determinant and adiabatic decomposition of the
  {$\zeta$}-determinant of the {D}irac {L}aplacian}, Lett. Math. Phys.
  \textbf{52} (2000), no.~4, 329--337.

\bibitem{PaJ-WoK02}
\bysame, \emph{Analytic surgery of the
  {$\zeta$}-determinant of the {D}irac operator}, Nuclear Phys. B Proc. Suppl.
  \textbf{104} (2002), 89--115, Quantum gravity and spectral geometry (Napoli,
  2001).

\bibitem{PaJ-WoK02II}
\bysame, \emph{Scattering theory and
  adiabatic decomposition of the {$\zeta$}-determinant of the {D}irac
  {L}aplacian}, Math. Res. Lett. \textbf{9} (2002), no.~1, 17--25.

\bibitem{PaJ-WoK05}
\bysame, \emph{Adiabatic decomposition of
  the {$\zeta$}-determinant and {D}irichlet to {N}eumann operator}, J. Geom.
  Phys. \textbf{55} (2005), no.~3, 241--266.

\bibitem{PaJ-WoK06}
\bysame, \emph{Adiabatic decomposition of the {$\zeta$}-determinant and
  scattering theory}, Michigan Math. J. \textbf{54} (2006), no.~1, 207--238.

\bibitem{PiP93}
P.~Piazza, \emph{On the index of elliptic operators on manifolds with
  boundary}, J. of Funct. Anal. \textbf{117} (1993), 308--359.

\bibitem{ray71-7-145}
D.B. Ray and I.M. Singer, \emph{R-torsion and the {L}aplacian on {R}iemannian
  manifolds}, Advances in Math. \textbf{7} (1971), 145--210.

\bibitem{ISiI88}
I.~M. Singer, \emph{The eta invariant and the index}, Mathematical aspects of
  string theory, World Scientific, Singapore, 1988, pp.~239--258.

\bibitem{scha06-73-042102}
M.~Schaden, \emph{{Comments on the Sign and Other Aspects of Semiclassical
  Casimir Energies}}, Phys. Rev. \textbf{A73} (2006), 042102.

\bibitem{vish95-167-1}
S.M. Vishik, \emph{Generalized ray-singer conjecture. i. a manifold with a
  smooth boundary}, Comm. Math. Phys. \textbf{167} (1995), 1--102.

\bibitem{weyl12-71-441}
H.~Weyl, \emph{Das asymptotische {V}erteilungsgesetz der {E}igenwerte linearer
  partieller {D}ifferentialgleichungen}, Math. Ann. \textbf{71} (1912),
  441--479.

\bibitem{WoK94}
K.~P. Wojciechowski, \emph{The additivity of the $\eta$-invariant: the
  case of an invertible tangential operator}, Houston J. Math. \textbf{20}
  (1994), no.~4, 603--621.

\bibitem{WoK95}
\bysame, \emph{The additivity of the $\eta$-invariant. {T}he case of a singular
  tangential operator}, Comm. Math. Phys. \textbf{169} (1995), no.~2, 315--327.

\end{thebibliography}

\def\cprime{$'$} \def\polhk#1{\setbox0=\hbox{#1}{\ooalign{\hidewidth
  \lower1.5ex\hbox{`}\hidewidth\crcr\unhbox0}}}
\providecommand{\bysame}{\leavevmode\hbox to3em{\hrulefill}\thinspace}
\providecommand{\MR}{\relax\ifhmode\unskip\space\fi MR }
\providecommand{\MRhref}[2]{%
  \href{http://www.ams.org/mathscinet-getitem?mr=#1}{#2}
}
\providecommand{\href}[2]{#2}

\end{document}